%% file: main.tex
\documentclass[preprint,12pt,3p]{elsarticle}
\usepackage{amssymb}
\usepackage{amsmath,amsfonts,amsthm}
\usepackage{comment}
\usepackage{natbib}
\usepackage[colorlinks=true,linkcolor=red,citecolor=blue,urlcolor=magenta]{hyperref}
\usepackage[ruled,linesnumbered,vlined]{algorithm2e}
\usepackage{booktabs}
\usepackage{comment}

\theoremstyle{plain}
\newtheorem{theorem}{Theorem}[section]
\newtheorem{lemma}[theorem]{Lemma}

\newtheorem{corollary}[theorem]{Corollary}
\theoremstyle{definition}
\newtheorem{definition}[theorem]{Definition}
\newtheorem{remark}[theorem]{Remark}
\newtheorem{assumption}[theorem]{Assumption}

\begin{document}

\begin{frontmatter}

\title{Applying Two-Grid Preconditioner for Subsurface Flow Simulation using Attention-enhanced Hybrid Network to Accelerate Multiscale Discretization in High-contrast Media} %% Article title

%% use optional labels to link authors explicitly to addresses:
%% \author[label1,label2]{}
%% \affiliation[label1]{organization={},
%%             addressline={},
%%             city={},
%%             postcode={},
%%             state={},
%%             country={}}
%%
%% \affiliation[label2]{organization={},
%%             addressline={},
%%             city={},
%%             postcode={},
%%             state={},
%%             country={}}

\author[xjtlu,uol]{Peiqi Li} %% Author name
\ead{Peiqi.Li23@student.xjtlu.edu.cn}
\author[xjtlu]{Jie Chen\corref{correspond}}
\ead{Jie.Chen01@xjtlu.edu.cn}
\author[eith]{Shibin Fu}
\ead{sfu@eitech.edu.cn}

\cortext[correspond]{Corresponding author.}

%% Author affiliation
\affiliation[xjtlu]{organization={Department of Applied Mathemaitcs, Xi'an Jiaotong-Liverpool University},
            postcode={215123}, 
            city={Suzhou},
            country={China}}
\affiliation[uol]{organization={Department of Mathematical Sciences, University of Liverpool},
                  postcode={L69 3BX},
                  city={Liverpool},
                  country={China}}
\affiliation[eith]{organization={Eastern Institute for Advanced Study, Eastern Institute of Technology},
                   postcode={315200},
                   city={Ningbo},
                   country={China}}

%% Abstract
\begin{abstract}
    In this paper, we study the efficient numerical solution of Darcy equations in strongly heterogeneous media with high-contrast permeability and propose a hybrid framework that combines learning with multiscale numerical methods. The learning component is used for the prediction of multiscale basis functions in the mixed generalized multiscale finite element method (mixed GMsFEM), with the goal of reducing the repeated local computations required in the offline stage. Once these basis functions are predicted, the global system is assembled and the pressure field is computed by a two-grid preconditioned solver. The resulting method accelerates the costly local basis-construction stage while retaining the multiscale discretization and preconditioned iterative structure of the underlying solver. Numerical experiments on two-dimensional heterogeneous Darcy problems show that the proposed framework yields more accurate final pressure reconstruction than several representative learning-based methods and remains stable under strong heterogeneity and high-contrast coefficients. In comparison with the traditional mixed GMsFEM, its main advantage lies in the efficiency of the basis-generation stage, while the quality of the global solve is still ensured by the two-grid preconditioner. These results indicate that accelerating multiscale basis construction through learning, while preserving a mature numerical solver for the global problem, provides a viable approach for high-resolution Darcy-type simulations.
\end{abstract}

%% Keywords

\end{frontmatter}

\input{section1_introduction}
\input{section2_preliminaries}
\input{section3_neural_operator}
\input{section4_gmsfem_two_grid_preconditioner}
\input{section_convergence}
\input{section5_numerical_experiments}
\input{section7_conclusion}

\input{sec_supp_info.tex}

\bibliographystyle{elsarticle-harv}
\bibliography{cas-refs}

\appendix
\newpage
\input{appendix_multiscale_solution}
\input{appendix_preconditioned_solutions}
\end{document}

%% file: section1_introduction.tex
\section{Introduction}
The Darcy equation is the fundamental mathematical model for describing the single-phase seepage process in porous media. It is widely applied in issues such as groundwater flow, oil and gas reservoir development, pollutant migration, and carbon dioxide sequestration. The actual underground medium often exhibits significant strong heterogeneity, multi-scale structure, and high contrast coefficient distribution, which causes the flow process to simultaneously contain local fine-scale transformations and global cross-scale coupling characteristics in space. To accurately depict such complex behaviors, numerical simulations usually require discretizing the governing equations on high-resolution grids. This step brings a significant computational burden of assembling and solving large-scale linear systems. Therefore, how to maintain the solution accuracy while reducing the computational cost caused by high resolution has always been one of the core issues in numerical simulation of porous medium flow \cite{bear2013dynamics}.

In response to these difficulties, multiscale methods provide an important framework for approximating and accelerating problems with strong heterogeneity in Darcy equations. In general, the core idea of these methods lies in effectively transferring complex information from the fine-scale medium to the coarser-scale calculations through local equivalence, basis function compression, or low-rank representation, thereby controlling the degree of freedom while retaining key physical characteristics. According to the different implementation approaches, the related methods can be roughly classified into several categories: One category is represented by homogenization and upscaling methods \cite{jivok2012homogenization,wen1996upsacling}, which embed the average effect of the fine-scale medium into the coarse-grid model by constructing equivalent permeability or equivalent transmission relationships in local regions; the other category is based on basis function construction strategies such as multiscale finite element method (MsFEM) and its extensions \cite{hou1997multiscale,efendiev2000convergence,efendiev2009multiscale,efendiev2007multiscale}. It generates multiscale basis functions containing fine-scale information through local boundary value problems or local spectral problems (LSP), and realizes global coupling in the coarse-scale space. Additionally, the model reduction method starts from the low-rank structure of the system response space and reduces the online solution cost of large-scale systems through snapshot sampling, basis vector extraction, and projection approximation \cite{quarteroni2014reduced,amsallem2011online}. For high heterogeneity and multi-scale Darcy flow problems, these methods alleviate the computational cost brought by direct solving of fine grids from different perspectives, but they also face different limitations, such as the limited resolution of equivalent parameter models in resolving local high-contrast details, and the insufficient adaptability of traditional reduction methods to parameter variations and complex local structures. In contrast, the generalized multiscale finite element method (GMsFEM) \cite{efendiev2013generalized,chung2014adaptive,chung2018constraint,chung2018generalized,chung2018fast} constructs snapshot spaces in local regions and further extracts dominant modes by using local spectral decomposition, thereby forming low-dimensional multiscale basis functions that can systematically represent the effects of fine-scale media. This shows stronger adaptability in high-contrast flow problems. On this basis, the mixed GMsFEM further combines multiscale basis functions with mixed finite element discretization, making the resulting velocity field satisfy local conservation properties, and thus is particularly suitable for underground flow simulation \cite{chung2015mixed}. A large number of studies have shown that these methods can better retain important fine-scale information on coarser grids and significantly improve the approximation quality of heterogeneous problems \cite{harder2013family,gregoire2005multiscale}. However, their bottlenecks are also clear: whether it is the local snapshot space solution, LSP solution, or the subsequent offline basis function generation, they all rely on repeated local computations. When the resolution increases, the number of coarse grid regions increases, or the number of multiscale basis functions to be retained increases, this part of the computational cost will significantly rise and gradually become the main computational burden in the entire framework.

Just as important as multiscale discretization is the efficient solution of the global linear system. For strongly heterogeneous and highly contrastive Darcy problems, even though the multiscale space has reduced the discretization scale to a certain extent, the algebraic system obtained by global assembly may still have unfavorable spectral properties, making the iterative solution process the main bottleneck in the overall computation. Therefore, the preconditioner is not an ancillary component in the subsequent implementation, but a core element in the numerical framework for this type of problem \cite{dohrmann2003preconditioner}. Specifically, the two-grid \cite{fu2024efficient,yang2019twogrid,fu2024adaptive,yang2022adaptive} or two-level \cite{ye2024robust,galvis2010domain,chung2013twolevel} method provides an effective and robust solution mechanism for Darcy-type systems in high-contrast media by combining coarse-space correction with local smoothing. For porous medium flow problems, existing studies have shown that an appropriately constructed coarse space can significantly improve the adaptability and convergence performance of iterative methods for highly heterogeneous systems. Thus, from the perspective of numerical methods, the multiscale space and the preconditioner correspond to the two interrelated and equally crucial aspects of approximate quality and solution efficiency.

In recent years, deep learning has provided a new perspective for solving partial differential equations (PDEs) \cite{li2026subsurface}. Neural operator methods \cite{lu2021learning,lu2022comprehensive,kovachki2023neural}, such as Fourier Neural Operators (FNO) \cite{li2020fourier}, directly learn the mapping from input to output, demonstrating strong expressiveness on parameterized PDEs like Darcy flow; convolutional networks like U-Net have advantages in multiscale feature extraction and local detail reconstruction \cite{ronneberger2015unet,liu2025learning}; physical-informed neural networks (PINN) \cite{rassi2019physics,li2024physics} embed the residual of physical equations into the loss function, making the network training subject to physical constraints. However, for strongly heterogeneous and highly contrastive Darcy problems, directly learning the final physical field still faces significant difficulties. On one hand, the local structures and sharp changes in multi-scale media are difficult to stably characterize through black-box mappings. On the other hand, residual-driven training is prone to optimization difficulties and accuracy degradation under high-contrast coefficient conditions. Thus, compared to directly applying learning methods to replace the entire numerical solution process, a more natural approach is to utilize learning methods to approximate the most expensive and structurally meaningful part of the numerical framework, such as the multi-scale basis functions in the GMsFEM framework \cite{chen2025prediction,li2025dual,choubineh2022innovative,choubineh2023deep}.

In this study, we proposed a hybrid solution framework for the Darcy equation. Instead of directly learning an end-to-end function space mapping, we focused on the computationally costly steps in the offline calculation, namely the offline multi-scale basis function calculation at the coarse grid scale. We developed a neural operator based on the frequency-domain attention mechanism as the global feature learning module, and used the improved U-Net based on the adaptive gating mechanism for further spatial refinement to obtain the multi-scale basis functions of the coarse grid. Subsequently, we will assemble based on the global sparse blocks and apply the two-grid preconditioner for iterative solution. We provided a rigorous convergence analysis for this hybrid framework, and experiments also demonstrated that our method significantly outperforms some well-known learning models in terms of inference efficiency and computational accuracy.

The organizational structure of the subsequent part of this article is as follows. \hyperref[sec:preliminaries]{Section~\ref{sec:preliminaries}} introduces the model problem and the construction of multiscale basis functions. \hyperref[sec:networks]{Section~\ref{sec:networks}} elaborates on the deep learning model and presents the construction method of the loss function. \hyperref[sec:method_preconditioner]{Section~\ref{sec:method_preconditioner}} introduces the two-grid preconditioner based on GMsFEN. \hyperref[sec:convergence]{Section~\ref{sec:convergence}} conducts convergence analysis for our hybrid framework. \hyperref[sec:experiments]{Section~\ref{sec:experiments}} provides detailed explanations on dataset generation, experimental environment, and demonstrates the model performance from multiple perspectives to showcase its superiority. In the \ref{appendix:ms_sol} and \ref{appendix:pre_sol}, we have provided detailed test cases based on different numbers of multiscale basis functions.

%% file: section2_preliminaries.tex
\section{Preliminaries}
\label{sec:preliminaries}
\begin{comment}
    Part 1. Model Problem: What we need to solve?
    Part 2. Construction of Multiscale Basis Function.
    Part 3. Two-Grid Preconditioner: Basic Definition.
\end{comment}
\subsection{Model Problem}
Consider the first-order problem of mixed form in the Lipschitz continuous domain $\Omega$:
\begin{equation}
    \begin{aligned}
        \kappa^{-1}v + \nabla p &= 0\quad \text{in}~\Omega, & (\text{Darcy's Law}) \\
        \nabla \cdot v &= f\quad \text{in}~\Omega, & (\text{Mass Conservation})
    \end{aligned}
\end{equation}
with non-homogeneous Neumann boundary condition $v \cdot n = g,~\text{on}~\partial\Omega$, where $\kappa$ is the given permeability field with high-contrast coefficient, $p$ and $v$ are the pressure and velocity fields, respectively. $f$ denotes the known source term, and $n$ is the outward unit vector of $\partial\Omega$.

Next we describe the multiscale finite element partition. Let $\mathcal{T}^H$ denote the non-overlapping conforming finite element partition of $\Omega$, and we call it coarse grid block. Then each coarse gird is partitioned into a connected union of fine grid blocks, where $H$ and $h$ are the mesh size of coarse and fine grids, respectively. Denote $\mathcal{E}^H := \bigcup_{i=1}^{N_e}\{E_i\}$ as the set of coarse edges, where $N_e$ is the number of coarse edges. Define the coarse neighborhood $w_i$ corresponding to the coarse edge $E_i$ as the union of all coarse-grid blocks having the edge $E_i$:
\begin{equation}
    w_i = \bigcup \{ K_j \in \mathcal{T}^H;~~E_i \in \partial K_j \}
\end{equation}
Similarly, we define the fine grid mesh as $t_i$, and $\Omega=\bigcup_{i=1}^{M_t}t_i$, where $M_t$ is the number of fine grid. Also $\mathcal{E}^h=\bigcup_{i=1}^{M_e}e_i$ is the set of fine edges. 

Define
\begin{equation}
    L^2(\Omega) := \left\{ v:~v\in\Omega~\text{and}~\int_{\Omega}v~dx<\infty \right\} \notag
\end{equation}
and
\begin{equation}
    H(\text{div},\Omega) := \left\{ v=(v_1,v_2)\in \left( L^2(\Omega) \right)^2,~\nabla \cdot v \in L^2(\Omega) \right\} \notag
\end{equation}

Then we define the multiscale finite element space
\begin{equation}
    \begin{aligned}
        V_h &= \left\{ 
        v_h\in V,~v_h|_t= 
        \begin{bmatrix}
            a_t \\ c_t  
        \end{bmatrix}
        +
        \left[ b_t,d_t \right]
        \begin{bmatrix}
            x_1 \\ x_2
        \end{bmatrix},~~
        a_t,b_t,c_t,d_t\in\mathbb{R}, t \in \mathcal{T}^h
        \right\} \\
        Q_h &= \left\{ q_h \in Q,~q_h~\text{is constant on each coarse element in }\mathcal{T}^h \right\}
    \end{aligned} \notag
\end{equation}
where $V=H(\text{div},\Omega)$ and $Q=L^2(\Omega)$.

For the coarse grids, define
\begin{equation}
    \begin{aligned}
        V_H :&
        = \bigoplus_{{\mathcal{E}^H}} \left\{  \Psi_i \right\},\quad \left\{ \Psi_i \right\}~\text{is the set of multiscale basis functions on}~E_i \\
        Q_H :&= \left\{ \text{piecewise constant on}~\mathcal{T}^H \right\} \notag
    \end{aligned}
\end{equation}
together with $V_H^0:=V_H \cap\left\{ v\in V_H,~v\cdot n=0~\text{on}~\partial\Omega \right\}$, i.e., $V_H^0=\bigoplus_{\mathcal{E}^h}\left\{ \Psi_i \right\}$.

Following the above definitions, the purpose of mixed GMsFEM is to find $\left( p_H,v_H \right)$ such that
\begin{equation}
    \begin{aligned}
        \int_{\Omega}\kappa^{-1}v_H\cdot w_H - \int_{\Omega}p_H\nabla\cdot w_H &= 0, & \forall~w_H\in V_H^0 \\
        -\int_{\Omega}q_H\nabla\cdot v_H &= -\int_\Omega fq_H, & \forall~q_H \in Q_H
    \end{aligned}
    \label{eq:coarse_weak}
\end{equation}
We call the solution that satisfies \eqref{eq:coarse_weak} multiscale solution or coarse-grid solution. Similarly we need to define the fine-grid solution as the reference to test the performance of multiscale solution:
\begin{equation}
    \begin{aligned}
        \int_\Omega \kappa^{-1}v_h\cdot w_h - \int_\Omega p_h\nabla\cdot w_h &= 0, & \forall~w_h\in V_h^0 \\
        -\int_\Omega q_h\nabla\cdot w_h &= -\int_\Omega fq_h, & \forall~q_h \in Q_h
    \end{aligned}
    \label{eq:fine_weak}
\end{equation}

Let $\left\{ \phi_1,\cdots,\phi_n \right\}$ and $\left\{ q_1,\cdots,q_m \right\}$ denote the basis set of $V_h$ and $Q_h$, respectively. Assume that $v_h=\sum_{i=1}^{n}v_i\phi_i$ and $p_h=\sum_{i=1}^m p_iq_i$, we can rewrite \eqref{eq:fine_weak} as the matrix representation
\begin{equation}
    \begin{bmatrix}
        M & B^T \\
        B & 0
    \end{bmatrix}
    \begin{bmatrix}
        v_h \\ p_h
    \end{bmatrix}
    =
    \begin{bmatrix}
        0 \\ F
    \end{bmatrix}
\end{equation}
where $M$ is a symmetric positive definite matrix with the entries $M_{ij}=\int_\Omega \kappa^{-1}\phi_i\phi_j$; $B$ is the approximation of the divergence operator with $B_{ij}=-\int_{\Omega}q_i \nabla\cdot \phi_j$; $F$ is the vector with $F_i = \int_\Omega fq_i$. By the velocity elimination technique, the trapezoidal rule to compute $M_{ij}$ can be transformed into a diagonal matrix, which can be easily inverted. Then one can obtain
\begin{equation}
    \left( BM^{-1}B^T \right)p_h = F
\end{equation}
where $BM^{-1}B^T$ is a symmetric and semipositive definite matrix.

\subsection{Construction of Multiscale Basis Functions for $Q_H$}
The first step to compute multiscale basis functions for pressure is to build the snapshot space, which refers to the local spatial domain or solution. In our problem setting, we construct the snapshot spaces by solving a series of local problems: for each $T_i$, find $(v_j^{(i)},p_j^{(i)})\in (V_h,Q_h)|_{T_i}$ such that
\begin{equation}
    \begin{aligned}
        \kappa^{-1}v_j^{(i)}+\nabla p_j^{(i)} &= 0, & \text{in}~T_i \\
        \nabla \cdot v_j^{(i)} &= 0, & \text{in}~T_i \\
        p_j^{(i)} &= \delta_j^{(i)}, & \text{on}~\partial T_i
    \end{aligned}
\end{equation}
Here the coarse edges can be represented as the union of fine grids ($\partial T_i = \bigcup_{j=1}^{J_i}e_j$), where $\delta_j^{(i)}$ is the piecewise constant function defined on $\partial T_i$ and satisfies
\begin{equation}
    p_j^{(i)}
    =
    \begin{cases}
        1, & \text{in}~e_i \\
        0, & \text{on the others edges of $\partial T_i$}
    \end{cases} \notag
\end{equation}

These Dirichlet boundary-value problems can be numerically solved using the lowest Raviart-Thomas element. Thus, the solution set of above questions on each coarse element can generate the snapshot space, which is defined as
\begin{equation}
    Q_{\text{snap}} := \text{span}\left\{ \Psi_j^{\text{snap}}, 1 \leq j \leq M_{\text{snap}},~M_{\text{snap}}=\sum_{i=1}^{N_t}J_i \right\} \notag
\end{equation}

Reorder the basis functions of snapshot space, one can obtain
\begin{equation}
    R_{\text{snap}}=\left[\Psi_1^{\text{snap}},\cdots,\Psi_{M_{\text{snap}}}^{\text{snap}} \right] \notag
\end{equation}
to project the snapshot space to fine-grid space.

Next step is to build the offline space. We further reduce the snapshot space. For each snapshot space $Q_{\text{snap}}^{(i)}$ of $T_i$, we solve local spectral problem
\begin{equation}
    A_{\text{snap}}^{(i)}\Psi_{\text{off}}^{(i)}=\lambda_{\text{off}}^{(i)}S_{\text{snap}}^{(i)}\Psi_{\text{snap}}^{(i)} \notag
\end{equation}
where
\begin{equation}
    \begin{aligned}
        A_{\text{snap}}^{(i)} &= a_i(\psi_m^{i,\text{snap}},\psi_n^{i,\text{snap}})=R_{\text{snap}}^TAR_\text{snap}\\
        S_{\text{snap}}^{(i)} &= s_i(\psi_m^{i,\text{snap}},\psi_n^{i,\text{snap}})=R_{\text{snap}}^TSR_{\text{snap}}
    \end{aligned} \notag
\end{equation}
, $a_i(p,w)=\sum_{e}\kappa [p][w],~s_i(p,w)=\int_{T_i}\kappa pw$, $[\cdot]$ is the jump and $e$ is is an interior fine edge in $T_i$; $A$ and $S$ are the stiffness and mass matrices of fine grid. Reorder the eigenvalues by
\begin{equation}
    \lambda_i^{(i)}<\cdots <\lambda_{J_i}^{(i)} \notag
\end{equation}
and choose the first $l_i$ eigenvalues with corresponding eigenvectors $Z_k^{(i)}=(Z_{kj}^{(i)})_{j=1}^{J_i}$. The multiscale basis functions of offline space can be constructed by the multiplication $\Psi_{k}^{i,\text{off}}=\sum_{j=1}^{J_i}Z_{kj}^{(i)}\Psi_j^{i,\text{snap}}$. Here we define $M_{\text{off}}=\sum_{i=1}^{N_t}l_i$, the offline space can be rewritten as
\begin{equation}
    Q_{\text{off}}=\text{span}\left\{ \Psi_k^{\text{off}},~1 \leq k \leq M_{\text{off}} \right\} \notag
\end{equation}

Similarly, we can formulate the project matrix
\begin{equation}
    R_{\text{off}}=\left[ \Psi_1^{\text{off}},\cdots,\Psi_{M_{\text{off}}}^{\text{off}} \right] \notag
\end{equation}

Therefore, the problem needed to be solved is to find $(p_H,v_H)$ such that
\begin{equation}
   \begin{bmatrix}
       M & B^TR_{\text{off}} \\
       R_{\text{off}}^TB & 0
   \end{bmatrix}
   \begin{bmatrix}
       v_H \\ p_H
   \end{bmatrix}
   =
   \begin{bmatrix}
       0 \\ R^T_{\text{off}}F
   \end{bmatrix}
\end{equation}
with the one without velocity term
\begin{equation}
    \left(R^T_{\text{off}}BM^{-1}B^TR_{\text{off}}\right)p_H=R_{\text{off}}F
\end{equation}

In practical applications, permeability fields exhibit characteristics such as high resolution, multi-scale, and strong heterogeneity. Even with the selection of an appropriate number of multi-scale basis functions, the offline computation phase still incurs costs. Based on this consideration, we propose using a hybrid neural network framework to replace this time-consuming stage. Once the basis functions are computed quickly, $R_{\text{off}}$ can be easily assembled.

%% file: section3_neural_operator.tex
\section{Attention-enhanced Hybrid Network to Construct Multiscale Basis Functions}
\label{sec:networks}
In this section, we discuss the network architecture used for quickly computing offline basis functions.

\begin{figure}[t]
    \centering
    \includegraphics[width=1.0\linewidth]{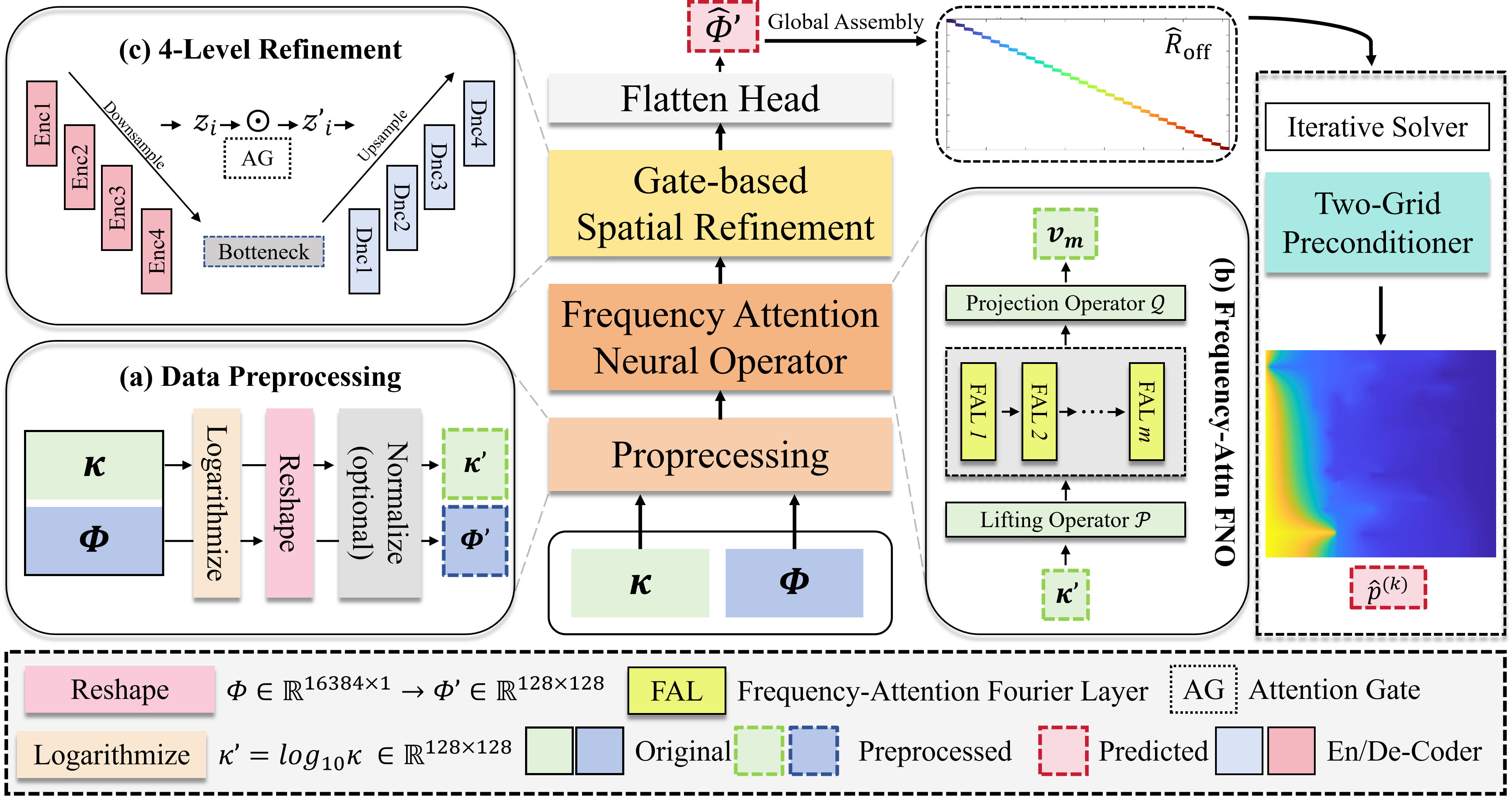}
    \caption{Schematic illustration of the overall technical pipeline of the proposed framework.}
    \label{fig:pipeline}
\end{figure}

\subsection{Network Architecture}

We developed a hybrid deep network for the prediction of multiscale basis functions. The core idea is to organically integrate the superior global operator learning capability of neural operators with the strong local detail recovery capability of spatial domain convolutional networks, forming a unified framework featuring global modeling in the frequency domain and progressive refinement in the spatial domain. Specifically, the model adopts a Fourier neural operator (FNO) incorporated with an attention mechanism to conduct global feature extraction and coarse prediction. Afterwards, an Attention U-Net at the backend acts as a spatial-domain refinement module to achieve progressive correction and detail enhancement for the coarse prediction outputs. Experimental results validate that this architecture enables effective learning of the nonlinear mapping from permeability fields to multiscale basis functions (see \hyperref[sec:experiments]{Section~\ref{sec:experiments}}).

The overall network can be abstracted as a composite mapping
\begin{equation}
    \mathcal{G}_{\Theta}:~\kappa \mapsto \Psi,\quad \Psi=\mathcal{R}_{\theta_r}\left( \mathcal{F}_{\theta_f}(\kappa) \right)
\end{equation}
where $\mathcal{F}_{\theta_f}$ denotes the fronted neural operator module, and $\mathcal{R}_{\theta_r}$ represents the backend spatial-domain refinement module.

To enhance the model's perception of spatial information in non-periodic priority regions, we explicitly introduce the coordinate encoding into the operator input. A normalized grid is defined as
\begin{equation}
    g(x,y):=\left( \xi(x),\eta(y) \right) \in \left[0,1\right]^2
\end{equation}
which is fed into the operator together with the permeability field. The input is then mapped to a high-dimensional space via a lifting operator:
\begin{equation}
    v_0=\mathcal{P}(\tilde{x})
\end{equation}
where $\mathcal{P}$ is the lifting operator (this step can be implemented by a linear fully-connected network).

In the frequency-domain propagation a 2-dimensional fast Fourier transform (2D-FFT) is applied to the input field:
\begin{equation}
    \hat{v}=\mathcal{F}(v_0)=\text{FFT}(v_0) \notag
\end{equation}

Instead of learning kernel operators across all frequencies, the model only retains a limited number of low-frequency modes to capture key long-range interactions. Let the retained frequency indices be $(\mu,\nu)\in \{1,\cdots,m_1 \}\times \{ 1,\cdots,m_2\}$, and the positive and negative frequency components undergo channel-wise mixing using independent complex weight tensors $W^{(+)}$ and $W^{(-)}$:
\begin{equation}
    \begin{aligned}
        \hat{y}^{(+)}_{b,o,\mu,\nu} &= \sum_{i=1}^{C}\hat{v}_{b,i,\mu,\nu}W^{(+)}_{i,o,\mu,\nu} \\
        \hat{y}^{(-)}_{b,o,\mu,\nu} &= \sum_{i=1}^{N}\hat{v}_{b,i,\mu,\nu}W^{(-)}_{i,o,\mu,\nu} 
    \end{aligned}\notag
\end{equation}

Essentially, we perform a complex linear transformation on each retained frequency and achieve a discrete approximation of the kernel integral along the channel dimension. Different from the standard FNO, we further introduce a frequency-wise attention mechanism. For each retained frequency point $(\mu,\nu)$, the network learns trainable weights $A^{(+)}_{o,\mu,\nu}$ and $A^{(-)}_{o,\mu,\nu}$ to adaptively adjust the importance different frequency components, yielding the frequency-domain output:
\begin{equation}
    \tilde{y}^{(+)}=A^{(+)}\hat{y}^{(+)},\quad \tilde{y}^{(-)}=A^{(-)}\hat{y}^{(-)} \notag
\end{equation}

\begin{remark}
For the normalization of attention weights, we consider two implementations of nonlinear activation functions. The first employs the sigmoid function: $A^{(\pm)}=\text{sigmoid}~(\alpha^{(\pm)})\in[0,1]^2$, under which the gating of each frequency component is independent of others. The second uses the softmax function: $A^{(\pm)}=\frac{\exp(\alpha^{(\pm)})}{\sum \exp(\alpha^{(\pm)})}$, which normalizes all retained frequencies jointly within the same channel, thereby forming competitive frequency allocation. In our numerical experiments, we adopt the first approach.
\end{remark}

After frequency-domain weighting, the model pads the retained modes back into the spectral tensor and transforms them back to the spatial domain via the inverse FFT (iFFT):
\begin{equation}
    y=\mathcal{F}^{-1}(\tilde{y})=i\text{FFT}(\tilde{y}) \notag
\end{equation}

To address the insufficient expression of local nonlinearity and fine-scale geometric details in pure frequency-domain linear propagation, each frequency-domain block is equipped with a parallel $1 \times 1$ spatial convolution branch $W_s$, which is fused with the frequency-domain branch:
\begin{equation}
    v_{l+1}=\sigma\left(\mathcal{K}_{l}(v_l)+W_{s,l}(v_l)\right)
\end{equation}
where $\mathcal{K}_l$ is the attention operator of the $l$-th layer, and $\sigma$ is the GELU activation.

The output of the frontend neural operator is a coarse prediction field that encapsulates global dependencies and dominant spectral information. However, for multiscale basis functions, modeling only global correlations is insufficient. Multiscale basis functions are typically influenced by multiple factors simultaneously (e.g., local medium jumps, boundary constraints, coarse mesh neighborhood geometry, and high-contrast channel connectivity). Therefore, we further perform multiscale hierarchical refinement in the spatial domain and design a 4-level UNet architecture (an encoder-decoder structure).

In the encoding stage, the double convolution module extracts local textures, edges, and block structures through successive local convolutions and nonlinear activations, while progressive downsampling expands the receptive field, enabling the network to further reorganize the coarse-scale responses provided by the neural operator in the spatial domain. The expanding path structure of the U-Net is inherently designed to balance contextual information and precise localization, which aligns well with the requirement of "global trend + local support structures" in basis function prediction.

In the decoding stage, the model follows the standard U-Net pipeline of upsampling–skip connection–convolution recovery, but instead of directly concatenating encoder features, it introduces an Attention Gate before each level of skip connection. Let the decoder gating feature be $g$ and the corresponding encoder skip feature be $x$; the gate first projects both to the same intermediate dimension:
\begin{equation}
    q_g=W_g * g,\quad q_x= W_x * x \notag
\end{equation}

Then the gating response is constructed:
\begin{equation}
    \psi=\sigma \left( W_{\psi} * \Phi(q_g+q_x) \right)
\end{equation}
where $\Phi$ represents a nonlinear activation, and $\sigma$ is the sigmoid function. Finally, the skip feature after gating selection is expressed as
\begin{equation}
    \tilde{x}=x \odot \psi
\end{equation}

This implies that local details from the encoder are not transmitted indiscriminately, but are spatially selectively enhanced according to the semantic requirements of the current decoding stage. For our problem, this mechanism is particularly crucial, as not all local high-frequency information is equally beneficial for basis function construction; the Attention Gate can suppress noise propagation from irrelevant regions and highlight detailed information related to the support regions of target basis functions, interface variations, and local anomalous structures. The core idea of Attention U-Net is precisely to enhance the effectiveness of skip connections through a gating mechanism and reduce the interference of irrelevant responses on the reconstruction results.

Therefore, the entire model can be interpreted as a two-stage collaborative approximation. In the first stage, neural operator learns the nonlocal operator mapping from the parameter field to the target field in function space, with the advantage of encoding long-range interactions, global trends, and cross-region coupling with low spectral mode complexity. In the second stage, Attention U-Net refines this coarse prediction in physical space, recovering boundaries, interfaces, and local fine-grained structures through multi-layer convolutions, context aggregation, and gated skip connections.

\subsection{Loss Function}
During the training phase, the network output is denoted as
\begin{equation}
    \hat{\Psi}=\mathcal{G}_{\Theta}(\kappa)(x)
\end{equation}
where $\kappa(x)$ denotes the permeability field, $\hat{\Psi}$ renotes the predicted multiscale basis functions, and $\Psi$ denotes the reference basis functions obtained via offline computation.

To simultaneously constrain the overall numerical accuracy and local spatial variation structure of the prediction results, this paper does not merely adopt the traditional mean squared error (MSE) loss, but constructs a joint loss function composed of field value reconstruction error and gradient consistency error. Its overall form is written as:
\begin{equation}
    \mathcal{L}_{\text{total}}=\mathcal{L}_{\text{data}}+\lambda_{\text{grad}}\mathcal{L}_{\text{grad}}
    \label{eq:total_loss}
\end{equation}
where $\lambda_{\text{grad}}>0$ is the weight of the gradient loss, used to balance the relative contributions of the field value fitting and local structure preservation.

First, the field value reconstruction term adopts the MSE form:
\begin{equation}
    \mathcal{L}_{\text{data}}=\frac{1}{N}\sum_{i=1}^N \left\| \hat{\Psi}^{(i)}-\Psi^{(i)} \right\|_2^2
    \label{eq:data_loss}
\end{equation}
Its role is to globally constrain the closeness between the predicted basis functions and the ground-truth basis functions in terms of amplitude, ensuring that the model learns the correct dominant distribution pattern and global numerical response.

However, for multiscale basis function prediction, relying solely on field value error is often insufficient to characterize their local geometric features. The reason is that multiscale basis functions require not only that the numerical amplitude approximates the ground-truth solution, but also that their spatial variation trends, local transition morphologies, and gradient distributions near interfaces can be reasonably recovered. Especially in high-contrast permeability fields, basis functions frequently exhibit strong directional variations and non-smooth transitions in local regions. If only point-wise error is minimized, the model may achieve a low global MSE but fail to accurately recover local boundaries, sharp transitions, or fine-scale support structures. Based on this consideration, this paper further introduces a gradient consistency loss to constrain the first-order discrete derivatives of the predicted and ground-truth fields in two spatial directions.

Define the first-order forward differences of the predicted field and reference field on the discrete grid as:
\begin{equation}
    \begin{aligned}
        D_x\hat{\Psi}_{i,j} &= \hat{\Psi}_{i+1,j}-\hat{\Psi}_{i,j},\quad D_y\hat{\Psi}_{i,j}=\hat{\Psi}_{i,j+1}-\hat{\Psi}_{i,j} \\
        D_x\Psi_{i,j} &= \Psi_{i+1,j}-\Psi_{i,j},\quad D_y \Psi_{i,j}=\Psi_{i,j+1}-\Psi_{i,j}
    \end{aligned}
    \notag
\end{equation}

Thus, the gradient consistency loss is defined as:
\begin{equation}
    \mathcal{L}_{\text{grad}}=\frac{1}{N}\sum_{i=1}^{N}\left(
    \left\| 
    D_x\hat{\Psi}^{(i)}-D_x\Psi^{(i)}
    \right\|_2^2
    +
    \left\|
    D_y\hat{\Psi}^{(i)}-D_y\Psi^{(i)}
    \right\|_2^2
    \right)
    \label{eq:grad_loss}
\end{equation}

Combining \eqref{eq:data_loss} and \eqref{eq:grad_loss}, the final training objective is
\begin{equation}
    \arg\min_{\Theta}~\mathcal{L}_{\text{data}}(\Psi,\hat{\Psi}_{\Theta})+\lambda_{\text{grad}}\mathcal{L}_{\text{grad}}(\Psi,\hat{\Psi}_{\Theta})
\end{equation}

Here, the first term ensures that the predicted basis functions approximate the ground-truth basis functions in terms of global amplitude and overall morphology, while the second term emphasizes the consistency of the local gradient field, directing the model to focus more on the recovery of boundaries, interfaces, and local fine-scale structures. With the addition of gradient loss, the model can better preserve the spatial structural features of basis functions, reduce the over-smoothing phenomenon that may occur when training with MSE alone, and thereby improve the usability and stability of the prediction results in subsequent numerical computation workflows.

%% file: section4_gmsfem_two_grid_preconditioner.tex
\section{GMsFEM-based Two-Grid Preconditioner}
\label{sec:method_preconditioner}

In this section, we illustrate the algorithm of two-grid preconditioner used in solving the pressure fields.

The preconditioner consists of two key components: smoother and coarse preconditioner. The smoother is used to remove high-frequency errors, and the coarse preconditioner is used to exchange global information. The step of computation can be viewed in \hyperref[alg:two_grid]{Algorithm~\ref{alg:two_grid}}.

\begin{remark}
    There are a lot of choices of the smoother, such as Jacobi iteration, Gauss-Seidel iteration, and the incomplete low upper (ILU) decomposition. In our experiment, we use ILU(0) as the smoother.
\end{remark}

\begin{remark}
    For the coarse preconditioner, we need to solve some coarse systems. The first way is to use the complete LU decomposition (robust but with huge memory cost) The second way is to use the iterative solver with ILU preconditioner (cheap but less robust).
\end{remark}

\begin{algorithm}[H]
    \caption{Two-grid preconditioner for solving $Ax = F$}
    \label{alg:two_grid}
    
    \KwIn{Initial guess $x^0$}
    \KwOut{Solution $x$}
    
    $x^1 \gets x^0 + S\bigl(F - Ax^0\bigr)$ \tcp*{Pre-smoothing, $S$ is a smoother}
    $r_1 \gets R\bigl(F - Ax^1\bigr)$ \tcp*{Compute coarse-grid residual}
    Solve $\bigl(RAR^T\bigr)x_c = r_1$ \tcp*{coarse-grid correction}
    $x^2 \gets x^1 + R^T x_c$ \tcp*{Project coarse-grid correction to fine grid}
    $x^3 \gets x^2 + S\bigl(F - Ax^2\bigr)$ \tcp*{Post-smoothing, $S$ is a smoother}
\end{algorithm}

\hyperref[fig:pipeline]{Figure~\ref{fig:pipeline}} illustrates the total workflow of our work, including data preprocessing (will be introduced in \hyperref[sec:experiments]{Section~\ref{sec:experiments}}), the structure of our hybrid network (\hyperref[sec:networks]{Section~\ref{sec:networks}}), and the two-grid algorithm.

%% file: section_convergence.tex
\section{Convergence Analysis}
\label{sec:convergence}

%%% 定义基函数误差
\begin{definition}[\textbf{Error of Multiscale Basis Functions}]
    Let $\psi_j$ and $\hat{\psi}_j$ denote the reference and predicted multiscale basis functions, then the total error can be defined as
    \begin{equation}
        \begin{aligned}
            \varepsilon^2_{\mathrm{basis}} &:= \sum_{j=1}^{M_t}
            \left\|
                \kappa^{1/2}(\psi_j-\hat{\psi}_j)
            \right\|_{L^2(\Omega)}^{2}, & \text{(Global)} \\ 
            \varepsilon_{i,l} &:= 
            \left\|
                \kappa^{1/2}(\psi_{i,l}-\hat{\psi}_{i,l})
            \right\|_{L^2(K_i)}, & \text{(Local)}
        \end{aligned}   
    \end{equation}
\end{definition}

%%% 定义能量范数和诱导算子范数
\begin{definition}[\textbf{Local Error Operator Norm}]
    For each coarse element $T_i$, let $E_i$ denote the local error propagation operator, the operator norm of $E_i$ is defined as
    \begin{equation}
        \left\| E_i \right\|_{l^2\rightarrow a_i} := \sup_{\alpha \neq 0}\frac{\|E_i\alpha\|_{a_i}}{\|\alpha\|_2}
    \end{equation}
    where $\| \cdot \|_{a_i}$ is the local energy norm $\|v\|_{a_i}=a_i(v,v)$.

    In this research, we assume that there exists $\varepsilon_{\mathrm{loc}}>0$ such that
    \begin{equation}
        \sum_{i=1}^{N_t}\left\| E_i \right\|_{l^2 \rightarrow a_i}^2 \leq \varepsilon_{\mathrm{loc}}^2
        \label{assump:error_operator}
    \end{equation}
\end{definition}

\begin{theorem}
    Let $R_{\mathrm{off}}$ and $\hat{R}_{\mathrm{off}}$ be the reference and predicted offline multiscale basis function matrix, respectively.
    There exists a constant $C_{\mathrm{asm}}$ such that
    \begin{equation}
        \left\| \Delta R \right\| := \left\| R_{\mathrm{off}}-\hat{R}_{\mathrm{off}} \right\| \lesssim
        \left(
        \sum_{i=1}^{M_t}\left\| E_i \right\|_{l^2\rightarrow a_i}^{2}
        \right)^{1/2} \lesssim \varepsilon_{\mathrm{loc}}
    \end{equation}
\end{theorem}
\begin{proof}
    Based on the global coupling method, the reference and prediction global offline matrices are obtained by sparse interpolation of local blocks, therefore we have
    \begin{equation}
        \Delta R = \sum_{i=1}^{M_t} \mathcal{G}_i(E_i)
    \end{equation}
    where $\mathcal{G}_i$ represents the assembly embedding operator from the $i$-th local block to the global matrix, $E_i$ is the local error operator. Since assembly only involves index embedding and sparse insertion, it does not involve adding new errors, but only the propagation of errors from local to global. Thus, there exists a constant $C>0$ such that
    \begin{equation}
        \left\| \mathcal{G}_i(E_i) \right\| \leq C\left\| E_i \right\|_{l^2\rightarrow a_i} \notag
    \end{equation}
    then we have the below relationship
    \begin{equation}
        \left\| \Delta R \right\| \leq \sum_{i=1}^{M_t} \left\| \mathcal{G}_i(E_i) \right\| \lesssim \sum_{i=1}^{M_t}\left\| E_i \right\|_{l^2 \rightarrow a_i}
    \end{equation}

    Applying the Cauchy-Schwarz inequality, we have
    \begin{equation}
        \sum_{i=1}^{M_t}\left\| E_i \right\|_{l^2 \rightarrow a_i} \leq M_t^{1/2} \left( \sum_{i=1}^{M_t} \left\|E_i \right\|_{l^2 \rightarrow a_i}^{2} \right)^{1/2} \lesssim \left( \sum_{i=1}^{M_t} \left\|E_i \right\|_{l^2 \rightarrow a_i}^{2} \right)^{1/2}
    \end{equation}

    Based on the assumption \eqref{assump:error_operator}, we can complete the proof.
\end{proof}

\begin{theorem}
    In the predicted offline space $\hat{R}_{\mathrm{off}}$, the Galerkin error is controlled by the optimized approximation error, i.e.,\
    \begin{equation}
        \left\| p_h - \hat{p}_H \right\|_a \leq \inf_{\chi \in \hat{Q}_{\mathrm{off}}}\left\| p_h - \chi \right\|_a
    \end{equation}
    where $p_H$ and $\hat{p}_H$ denote the Galerkin approximation on $R_{\mathrm{off}}$ and $\hat{R}_{\mathrm{off}}$, respectively.
    \label{theo:cea_estimate}
\end{theorem}

\begin{proof}
    According to the Galerkin orthogonality, one can obtain
    \begin{equation}
        \begin{aligned}
            a(p_h,q) &= (f,q), & \forall~q \in Q_h \\
            a(\hat{p}_H, \chi) &= (f,\chi), & \forall \chi \in \hat{Q}_{\mathrm{off}} \\
            a(p_h-\hat{p}_H, \chi) &= 0, & \forall \chi \in \hat{Q}_{\mathrm{off}}
        \end{aligned}
    \end{equation}

    Hence, for all $\chi \in \hat{Q}_{\mathrm{off}}$, we have
    \begin{equation}
        a(p_h-\hat{p}_H, p_h-\hat{p}_H) = a(p_h-\hat{p}_H,p_h-\chi)
    \end{equation}

    Applying the Cauchy-Schwarz inequality, one has
    \begin{equation}
        \left\| p_h-\hat{p}_H \right\|_a^2 \leq \left\| p_h-\hat{p}_H \right\|_a \left\| p_h - \chi \right\|_a
    \end{equation}

    If the LHS is 0, then the conclusion is obviously true. Otherwise, canceling $ \left\| p_h-\hat{p}_H \right\|_a$ from both sides and take the lower bound of all $\chi$, one can obtain the conclusion, too.
\end{proof}

Next we will show that the approximation property of $\hat{Q}_{\mathrm{off}}$ is controlled by the true approximation and basis function errors. Before the proof, we need to introduce the approximation property of true offline space.

\begin{assumption}
    We assume that $\hat{Q}_{\mathrm{off}}$ has the same dimension with $Q_{\mathrm{off}}$, and there exists a constant $0 \leq c_R \leq C_R$ such that
    \begin{equation}
        c_R\|\alpha\|_2^2 \leq \left\| \kappa^{1/2}\sum_{j=1}^{N}\alpha_j\hat{\psi}_j \right\|_{L^2(\Omega)}^{2} \leq C_R\|\alpha\|_2^2,\quad \forall~\alpha \in \mathbb{R}^{N}
    \end{equation}
    \label{assump:dimension_consist}
\end{assumption}

\begin{lemma}[\textbf{Approximation Property of $Q_{\mathrm{off}}$} \cite{fu2024efficient}]
    There exists a constant $C_1>1$ independent of $\kappa$, $h$, and $H$ such that
    \begin{equation}
        \inf_{\chi \in Q_{\mathrm{off}}} \left\| \kappa^{1/2} \left( q-\chi \right) \right\|_{L^2(\Omega)}^2 \leq C_1\lambda_A^{-1}a(q,q),\quad \forall~q \in Q_h
        \label{eq:approximation_true}
    \end{equation}
    where $\lambda_A$ is the largest eigenvalue of $A=BM^{-1}B^T$.
    \label{lem:approx_property}
\end{lemma}

Our objective is to prove that this property is also true for the predicted space, i.e., \eqref{eq:approximation_true} is also true for $\hat{\chi} \in \hat{Q}_{\mathrm{off}}$.

\begin{theorem}
    \label{theo:convergence_two_grid}
    There exists a constant $\hat{C}_1>0$ such that
    \begin{equation}
        \inf_{\hat{\chi} \in \hat{Q}_{\mathrm{off}}}
        \left\| \kappa^{1/2}\left(q-\hat{\chi} \right) \right\|_{L^2(\Omega)}^2 \leq \hat{C}_1\lambda_A^{-1}a(q,q),\quad \forall~q \in Q_h
    \end{equation}
    and 
    \begin{equation}
        \hat{C}_1 =C_1 + 2CC_1^{1/2}\lambda_A^{1/2}\varepsilon_{\mathrm{basis}}+C^2\lambda_A\varepsilon_{\mathrm{basis}}^2
    \end{equation}
\end{theorem}

\begin{proof}
    According to \hyperref[lem:approx_property]{Lemma~\ref{lem:approx_property}}, for all $q \in Q_h$, there exists a $\chi^* \in Q_{\mathrm{off}}$ that satisfies
    \begin{equation}
        \left\| \kappa^{1/2}(q-\chi^*)
        \right\|_{L^2(\Omega)}^2
        \leq
        C_1\lambda_A^{-1}a(q,q)
    \end{equation}

    Let $\hat{\chi} \in \hat{Q}_{\mathrm{off}}$, rewrite $\chi^*$ and $\hat{\chi}$ as the linear combination of basis functions:
    \begin{equation}
        \chi^*=\sum_{j=1}^{N}p_j\psi_j,\quad \hat{\chi}=\sum_{j=1}^{N}p_j\hat{\psi}_j \notag
    \end{equation}
    Split $q-\hat{\chi}$ and take the L2 norm:
    \begin{equation}
        \begin{aligned}
            \left\| \kappa^{1/2}(q-\hat{\chi}) \right\|_{L^2(\Omega)} &= \left\| \kappa^{1/2}\left( (q-\chi^*)+(\chi^*-\hat{\chi}) \right) \right\|_{L^2(\Omega)} \\
            &\leq \underbrace{\left\| \kappa^{1/2}(q-\chi^*) \right\|_{L^2(\Omega)}}_{\mathrm{Term~1}} + 
            \underbrace{\left\| \kappa^{1/2}(\chi^*-\hat{\chi}) \right\|_{L^2(\Omega)}}_{\mathrm{Term~2}} & \text{(Triangle Inequality)}
        \end{aligned}
    \end{equation}

    For Term 2,
    \begin{equation}
        \begin{aligned}
            \left\| \kappa^{1/2}(\chi^*-\hat{\chi}) \right\|_{L^2(\Omega)} &= 
            \left\| \kappa^{1/2}\left( \sum_{j=1}^{N}p_j(\psi_j-\hat{\psi}_j) \right) \right\|_{L^2(\Omega)} \\ 
            &\leq
            \left\| p \right\|_2 \left( \sum_{j=1}^{N}\left\| \kappa^{1/2}(\psi_j-\hat{\psi}_j)\right\|_{L^2(\Omega)}^{2} \right)^{1/2} & \text{(Cauchy-Schwarz)} \\ 
            &= \left\| p \right\|_2 \varepsilon_{\mathrm{basis}}
        \end{aligned}
    \end{equation}

    By \hyperref[assump:dimension_consist]{Assumption~\ref{assump:dimension_consist}}, the coefficient vector $\|p\|_2$ is controlled by $\|\kappa^{1/2}\hat{\chi}\|_{L^2(\Omega)}$, and $\hat{\chi}$ can be controlled by the energy of $q$, thus we have
    \begin{equation}
        \| p \|_2 \leq Ca(q,q)^{1/2},~C~\text{is a constant}
    \end{equation}
    and
    \begin{equation}
        \left\| \kappa^{1/2}(q-\chi^*) \right\|_{L^2(\Omega)} \lesssim \varepsilon_{\mathrm{basis}}a(q,q)^{1/2}
        \label{eq:approx_term1}
    \end{equation}

    The Term 1 can be easily proved using \hyperref[lem:approx_property]{Lemma~\ref{lem:approx_property}}:
    \begin{equation}
        \left\| \kappa^{1/2}(q-\chi^*) \right\|_{L^2(\Omega)}
        \leq
        C_1^{1/2}\lambda_A^{-1/2}a(q,q)^{1/2}
        \label{eq:approx_term2}
    \end{equation}

    Combining \eqref{eq:approx_term1} and \eqref{eq:approx_term2}, one can obtain
    \begin{equation}
        \left\|
            \kappa^{1/2}(q-\hat{\chi})
        \right\|_{L^2(\Omega)}
        \leq
        C_1^{1/2}\lambda_A^{-1/2}a(q,q)^{1/2}+C\varepsilon_{\mathrm{basis}}a(q,q)^{1/2}
    \end{equation}  

    Square the terms on both sides of the equation:
    \begin{equation}
    \begin{aligned}
        \left\| \kappa^{1/2}(q-\chi^*) \right\|_{L^2(\Omega)}^2 &\leq \left( C_1\lambda_A^{-1}+2C_1C\lambda_A^{-1/2}\varepsilon_{\mathrm{basis}}+C^2\varepsilon_{\mathrm{basis}}^2\right)a(q,q) \\
        &=
        \left( C_1 + 2C1C\lambda_A^{1/2}\varepsilon_{\mathrm{basis}}+C^2\lambda_A\varepsilon_{\mathrm{basis}}^2 \right)\lambda_A^{-1}a(q,q) \\
        &=
        \hat{C}_1\lambda_A^{-1}a(q,q)
    \end{aligned}
    \end{equation}

    In particular, if $\varepsilon_{\mathrm{basis}}$ is sufficiently small, the linear term can be absorbed into the perturbation constant, and one obtains the simplified estimate
    \begin{equation}
        \hat{C}_1 \leq C_1 + C_{\mathrm{pert}}\varepsilon_{\mathrm{basis}}^2
    \end{equation}
\end{proof}

Next, we discuss the error propagation operator. Using the definition below:
\begin{definition}[\textbf{Error Propagation Operators}]
    Let $A=H-N$, where $H=(L+D)D^{-1}(L+D)^{T}$, $N$ is the rest matrix, D is the diagonal matrix, and $L$ is a strictly lower triangular part of A. Then define $K_1=I-H^{-1}A$ as the error propagation part of smoother. The error propagation operator of two-grid algorithm is
    \begin{equation}
        E_1 = K_1(I-P_0+E_0P_0)K_1
    \end{equation}
    where $E_0$ is the error propagation operator at the coarse-grid level, $P_0$ is the interpolation operator from $Q_h$ to $Q_{\mathrm{off}}$. Here we use the hat symbol to denote the ones on the predicted spaces. Define the norm
    \begin{equation}
        \begin{aligned}
            \| E_1 \|_A &:= \sup_{q \in Q_h} \frac{a(E_1q,q)}{a(q,q)} \\
            \| E_0 \|_A &:= \sup_{q \in Q_h} \frac{a(E_0q,q)}{a(q,q)}
        \end{aligned}
    \end{equation}
    then we also have $\hat{E}_1$ and $\hat{E_0}$ using the same definition.
\end{definition}

\begin{lemma}
    $E_1$ is the error propagation operation using the above definition; therefore, the estimate is satisfied:
    \begin{equation}
        \| E_1 \|_A \leq \eta \|E_0\|_A+\eta
    \end{equation}
    where $\eta=1-\frac{1}{C_1}$. When $\|E_0\|_A<1$, $\|E_1\|_A<1$ is satisfied.
    \label{lem:error_prop}
\end{lemma}

\begin{theorem}
    \label{theo:E1_error}
    Let $\hat{E}_0$ and $\hat{E}_1$ be the error propagation operators on predicted spaces, then the inequality is satisfied:
    \begin{equation}
        \| \hat{E}_1 \|_A \leq (1-\hat{\eta})\|\hat{E}_0\|_A+\hat{\eta}
    \end{equation}
    then one can have
    \begin{equation}
        \| \hat{E}_1 \|_A \lesssim 1-\frac{1}{C_1+C_{\mathrm{pert}}\varepsilon^2_{\mathrm{basis}}}
    \end{equation}
    If $\hat{E}_0\|_A<1$, then $\|\hat{E}_1\|_A<1$.
\end{theorem}

\begin{proof}
    According to \hyperref[lem:error_prop]{Lemma~\ref{lem:error_prop}}, the convergence factor of two-grid preconditioner is controlled by $C_1$. Here we are using the same smoothing operator $K_1$, so the smoothing property remains. The difference is that the space and interpolation operator change from $Q_{\mathrm{off}}$, $P_0$ to $\hat{Q}_{\mathrm{off}}$ and $\hat{P}_0$. According to \hyperref[theo:convergence_two_grid]{Theorem~\ref{theo:convergence_two_grid}}, $\hat{C}_1$ is controlled by $C_1+C_{\mathrm{pert}}\varepsilon_{\mathrm{basis}}^2$, so we can easily get
    \begin{equation}
        \|\hat{E}_1\|_A \leq (1-\hat{\eta})\|\hat{E}_0\|_A+\hat{\eta},\quad\hat{\eta}=1-\frac{1}{\hat{C}_1}
    \end{equation}
\end{proof}

\begin{theorem}
    Define the stiffness matrices on coarse-grid level
    \begin{equation}
        A_H=R_{\mathrm{off}}^TAR_{\mathrm{off}},\quad \hat{A}_H=\hat{R}_{\mathrm{off}}^TA\hat{R}_{\mathrm{off}} \notag
    \end{equation}

    So we have
    \begin{equation}
        \| \hat{A}_H-A_H \| \lesssim \|\Delta R\|+\|\Delta R\|^2 \lesssim \varepsilon_{\mathrm{loc}}+\varepsilon_{\mathrm{loc}}^2
    \end{equation}
\end{theorem}

\begin{proof}
    By the definition,
    \begin{equation}
        \begin{aligned}
            \hat{A}_H &= \hat{R}^T_{\mathrm{off}}A\hat{R}_{\mathrm{off}} = (R_{\mathrm{off}}+\Delta R)^TA(R_{\mathrm{off}}+\Delta R) \\
            \Delta A &= \Delta R^TAR_{\mathrm{off}}+R^T_{\mathrm{off}}A\Delta R + \Delta R^TA\Delta R
        \end{aligned}
    \end{equation}

    Take the matrix norm and apply the inequality, one can obtain
    \begin{equation}
        \begin{aligned}
            \|\Delta R^TAR_{\mathrm{off}}\| &\leq \|A\|\|R_{\mathrm{off}}\|\|\Delta R\| \\
            \|R_{\mathrm{off}}^TA\Delta R\| &\leq \|A\|\|R_{\mathrm{off}}\|\|\Delta R\| \\
            \|\Delta R^TA\Delta R\| &\leq \|A\|\|\Delta R\|^2
        \end{aligned}\notag
    \end{equation}
    Sum up these terms and we can obtain
    \begin{equation}
        \| \hat{A}_H - A_H\| \lesssim \|\Delta R\|+\|\Delta R\|^2 \lesssim \varepsilon_{\mathrm{loc}}+\varepsilon_{\mathrm{loc}}^2
    \end{equation}
\end{proof}

Therefore, we can derive the error convergence of pressure solutions using iterative two-grid preconditioner.
\begin{theorem}
    The final pressure error is jointly controlled by spatial approximation error, learning error and iteration error. Specifically, let $p^{(k)}$ denote the approximation solution after $k$ iterations of two-grid algorithm using $\hat{R}_{\mathrm{off}}$, then we have
    \begin{equation}
        \left\| p_h-p^{(k)} \right\|_a \lesssim \inf_{\chi \in Q_{\mathrm{off}}} \left\| p_h-\chi\right\|_a + \varepsilon_{\mathrm{basis}} + \|\hat{E}_1\|_A^k\|\hat{p}_H-p^{(0)}\|_a
    \end{equation}
    \label{theo:final_error}
\end{theorem}

\begin{proof}
    Take the error decomposition
    \begin{equation}
        p_h - p^{(k)} = \underbrace{(p_h-\hat{p}_H)}_{\mathrm{Term~1}} + \underbrace{(\hat{p}_H - p^{(k)})}_{\mathrm{Term~2}}
    \end{equation}

    For term 1, the error is controlled by the optimized approximation error of $\hat{Q}_{\mathrm{off}}$ (\hyperref[theo:cea_estimate]{Theorem~\ref{theo:cea_estimate}}). We further apply \hyperref[theo:convergence_two_grid]{Theorem~\ref{theo:convergence_two_grid}}, this approximation error is jointly controlled by the one in $Q_{\mathrm{off}}$ and $\varepsilon_{\mathrm{basis}}$. So we have
    \begin{equation}
        \| p_h-\hat{p}_H\|_a \lesssim \inf_{\chi \in Q_{\mathrm{off}}}\|p_h-\chi\|_a+\varepsilon_{\mathrm{basis}}
    \end{equation}

    The second term can directly be obtained by the definition of error propagation operator:
    \begin{equation}
        \|\hat{p}_H-p^{(k)}\|_a \leq \|\hat{E}_1\|_A^k\|\hat{p}_H-p^{(0)}\|_a
    \end{equation}

    Sum up the two term and we can obtain the conclusion.
\end{proof}

Furthermore, if we write the learning error $\varepsilon_{\mathrm{basis}}$ in detail, we will have another corollary:
\begin{corollary}
    The learning error can be decomposed as the sum of three types of sub-errors:
    \begin{enumerate}
        \item approximation error: $\varepsilon_{\mathrm{approx}}$;
        \item optimization error: $\varepsilon_{\mathrm{optim}}$;
        \item generalization error: $\varepsilon_{\mathrm{gen}}$.
    \end{enumerate}

    Substitute them into \hyperref[theo:final_error]{Theorem~\ref{theo:final_error}}, one can obtain
    \begin{equation}
        \| p_h-p^{(k)}\|_a \lesssim \inf_{\chi \in Q_{\mathrm{off}}}\|p_h-\chi\|_a + \varepsilon_{\mathrm{approx}}+\varepsilon_{\mathrm{optim}}+\varepsilon_{\mathrm{gen}}+\|\hat{E}_1\|_A^k\|\hat{p}_H-p^{(0)}\|_a
    \end{equation}
    and 
    \begin{equation}
        \hat{C}_1 \leq C_1 +  C_{\mathrm{pert}}(\varepsilon_{\mathrm{approx}}+\varepsilon_{\mathrm{optim}}+\varepsilon_{\mathrm{gen}})^2
    \end{equation}

    Substituting these into \hyperref[theo:E1_error]{Theorem~\ref{theo:E1_error}}, we know that $\|\hat{E}_1\|_A$ is controlled by the three types of neural network error. Thus we can obtain the final conclusion.
\end{corollary}

%% file: section5_numerical_experiments.tex
\section{Numerical Experiments}
\label{sec:experiments}

\subsection{Data Acquisition and Processing}
\begin{figure}[ht]
    \centering
    \includegraphics[width=1.0\linewidth]{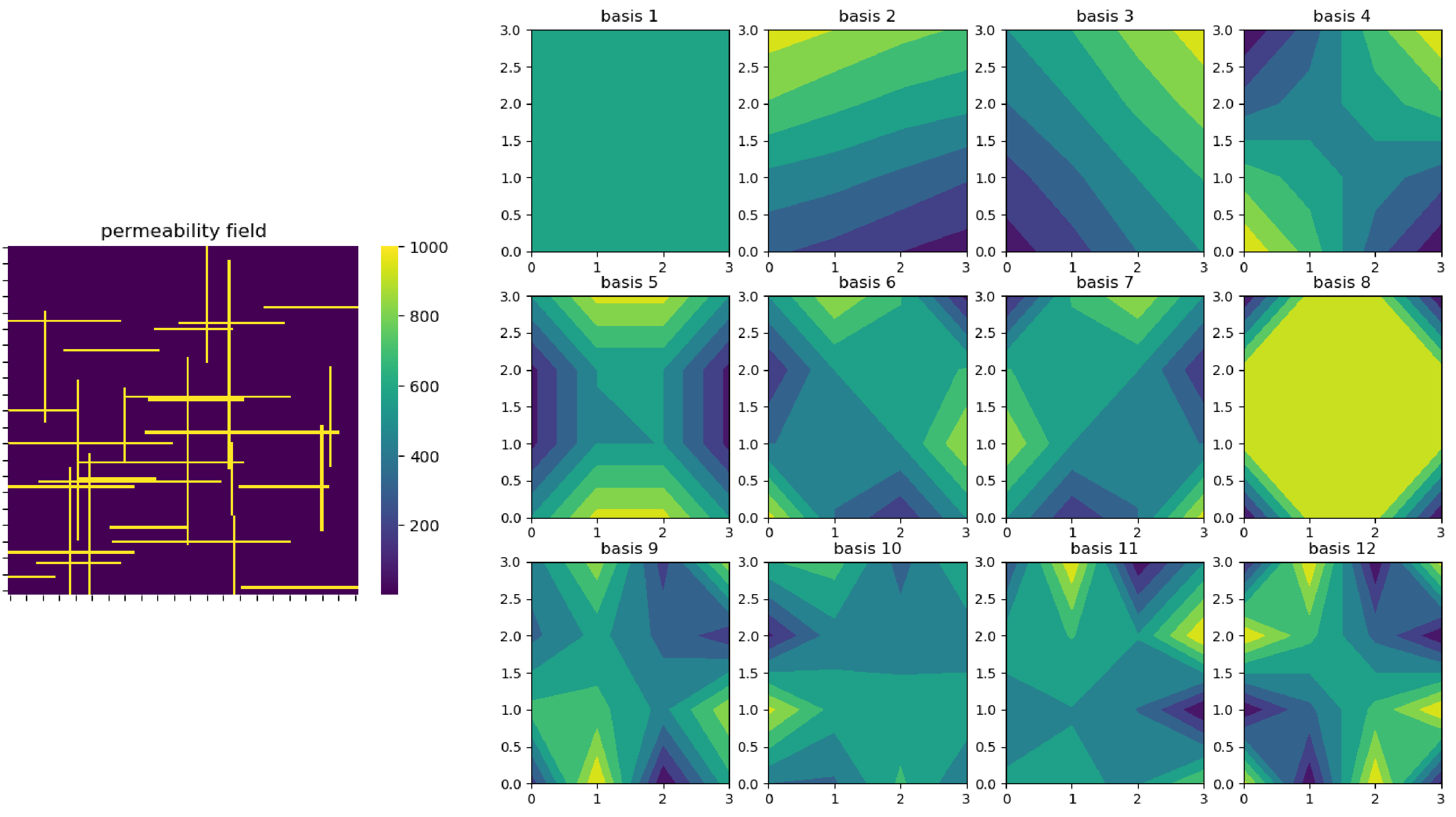}
    \caption{Example of our dataset. Right: high-contrast fractured permeability field. The yellow lines refer to the fractures. Right: total multiscale basis functions. Totally we have 12 basis functions, and the first (basis 1) is piecewise constant (1 or -1).}
    \label{fig:example_dataset}
\end{figure}

In this subsection, we illustrate the method of data generation, the information about the dataset, and the data preprocessing techniques.

The data generation is divided into two parts. First, we use Karhunen-Lo\`eve Expansion (KLE) to repeatedly generate the fractured permeability fields \cite{fukunaga1970application}. KLE is a useful tool for simulating random fields and has been used in many studies on the topic of subsurface flow modeling \cite{liu2025learning,chen2025prediction,li2025hybrid,li2025dual}.

The second part is to derive different types of pressure solutions, including reference solution, multiscale solution (see \hyperref[appendix:ms_sol]{\ref{appendix:ms_sol}}), and the predicted solutions using iterative two-grid algorithm (see this section and \hyperref[appendix:pre_sol]{\ref{appendix:pre_sol}}).

In the data preprocessing step, we logarithmize the $\kappa$, scaling down the original high-contrast fields to a small range. We can also optionally further normalize this field, but we did not choose to do so in our numerical experiments. For $\Psi$, we reshape the vector formed by vertically concatenating multiple coarse grid multi-scale basis functions into a 128*128 quasi-image. The data of pressure fields and offline multiscale basis functions are obtained using GMsFEM.

We generated a total of 4000 sets of unique random data, and randomly split 3500 sets of samples for model training and 500 sets for testing model performance. \hyperref[fig:example_dataset]{Figure~\ref{fig:example_dataset}} is one example of our data that will be used in model training.

\subsection{Experiment Setup}
The process of generating our data was carried out using MATLAB R2021a under the support of the Windows Server 2016 Datacenter system. This part involves the generation of permeability fields based on KLE, as well as the calculation of multi-scale basis functions, reference pressure solutions, multi-scale solutions, and two-grid iterative solution data. The training of the neural network was conducted using the Pytorch framework with an Nvidia Tesla V100 GPU of 32GB.

The network parameters were optimized using the Adam optimizer, which provides adaptive parameter updates and has demonstrated robust convergence behavior in our experiments \cite{kingma2014adam}.

The relevant hyperparameters are shown in \hyperref[tab:hyper_params]{Table~\ref{tab:hyper_params}}:

\begin{table}[htbp]
  \centering
  \caption{Experimental Hyperparameter Settings}
  \label{tab:hyper_params}
  \begin{tabular}{l|l}
    \toprule
    \textbf{Hyperparameter} & \textbf{Value} \\
    \midrule
    Truncated modes of Fourier operator & 16 (x-axis); 16 (y-axis) \\
    Fine grid resolution & $128\times 128$ \\
    Coarse grid resolution & $4\times 4$ \\
    Gradient loss weight & 0.1 \\
    Optimizer & Adam \\
    Learning rate & 0.001 \\
    Epochs & 60 \\
    \bottomrule
  \end{tabular}
\end{table}

\subsection{Experiment Results}
% 几个部分：
% part 1. 神经网络的训练结果，把学习曲线、MSE和R2的值弄成表格展示
% part 2. baseline的结果，需要展示的是：重构的结果+相对L2误差直方图+数据图表
% part 3. 基于mm=6的基函数等值线图+Roff+A+压力+绝对误差，然后说一下我们提供了其他mm选择时的数值算例在附录
\begin{comment}
    我们对结果的展示将会分成多个部分。首先，我们讨论混合神经网络在计算多尺度基函数中展现的效率。对于模型性能的评估，我们遵循经典，同时也是最常用的评估指标：均方误差（MSE）和决定系数（R2）。MSE用于衡量预测结果与参考解之间平方误差的平均水平，而R2用于评估模型对目标数据方差的解释能力。对应的计算方法为
    《两个指标的计算公式+符号解释》
    根据GMsFEM的理论，第一个多尺度基函数一定是分片的1或者-1，这一变量仅和空间位置有关，我们不对这一函数进行学习。图3为我们的方法计算得出的某一粗网格上11个多尺度基函数的示意图（basis2~12），同时，表2展示了在测试集上对不同基函数学习的MSE和R2.我们可以轻易观察到，所有的MSE都保持在10^-3范畴，R2都大于0.9。这些证据都能够说明我们的模型能够比较精确地学习到这一映射关系。图4为对应的学习曲线，这也能够侧面作证我们的方法在加速替代离线计算上具备良好的收敛能力。
\end{comment}
The experimental results are presented from several perspectives. We first focus on the efficiency and accuracy of the proposed hybrid neural network in computing multiscale basis functions. To evaluate the model performance, we adopt two classical and widely used metrics, namely the mean squared error (MSE) and the coefficient of determination ($R^2$). Specifically, MSE is used to measure the average squared difference between the predicted results and the reference solutions, while $R^2$ is used to assess the explanatory capability of the model with respect to the variance of the target data. These two metrics are defined as
\begin{equation}
    \mathrm{MSE}=\frac{1}{N}\sum_{i=1}^{N}\left( \psi_i - \hat{\psi}_i \right)^2, \quad R^2=1-\frac{\sum_{i=1}^{N}\left( \psi_i-\hat{\psi}_i \right)^2}{\sum_{i=1}^{N}\left( \psi_i-\bar{\psi}_i \right)^2}
\end{equation}
where $N$ denotes the total number of sample points, $\psi_i$ is the reference value of the $i$-th sample, $\hat{\psi}_i$ is the corresponding predicted value, and $\bar{y}$ represents the mean of all reference values. In general, a smaller MSE and a larger $R^2$ indicate better predictive performance. \hyperref[fig:basis_pred]{Figure~\ref{fig:basis_pred}} illustrates the 11 multiscale basis functions computed by the proposed method on a representative coarse block, namely basis 2–12.

According to the theory of GMsFEM, the first multiscale basis function is always a piecewise constant function taking the value 1 or -1. Since this basis function depends only on the spatial location, it is not learned in our framework. \hyperref[tab:performance_metrics]{Table~\ref{tab:performance_metrics}} reports the MSE and $R^2$ values of different basis functions on the test set. It can be clearly observed that all MSE values remain at the order of $10^{-3}$, while all $R^2$ values are above 0.9. These results demonstrate that the proposed model is able to accurately learn the mapping from the input permeability field to the corresponding multiscale basis functions. In addition, the learning curves shown in \hyperref[fig:learning_curve]{Figure~\ref{fig:learning_curve}} further verify the favorable convergence behavior of our method, indicating its strong potential for accelerating and replacing the conventional offline computation procedure.

\begin{figure}[t]
    \centering
    \includegraphics[width=1.0\linewidth]{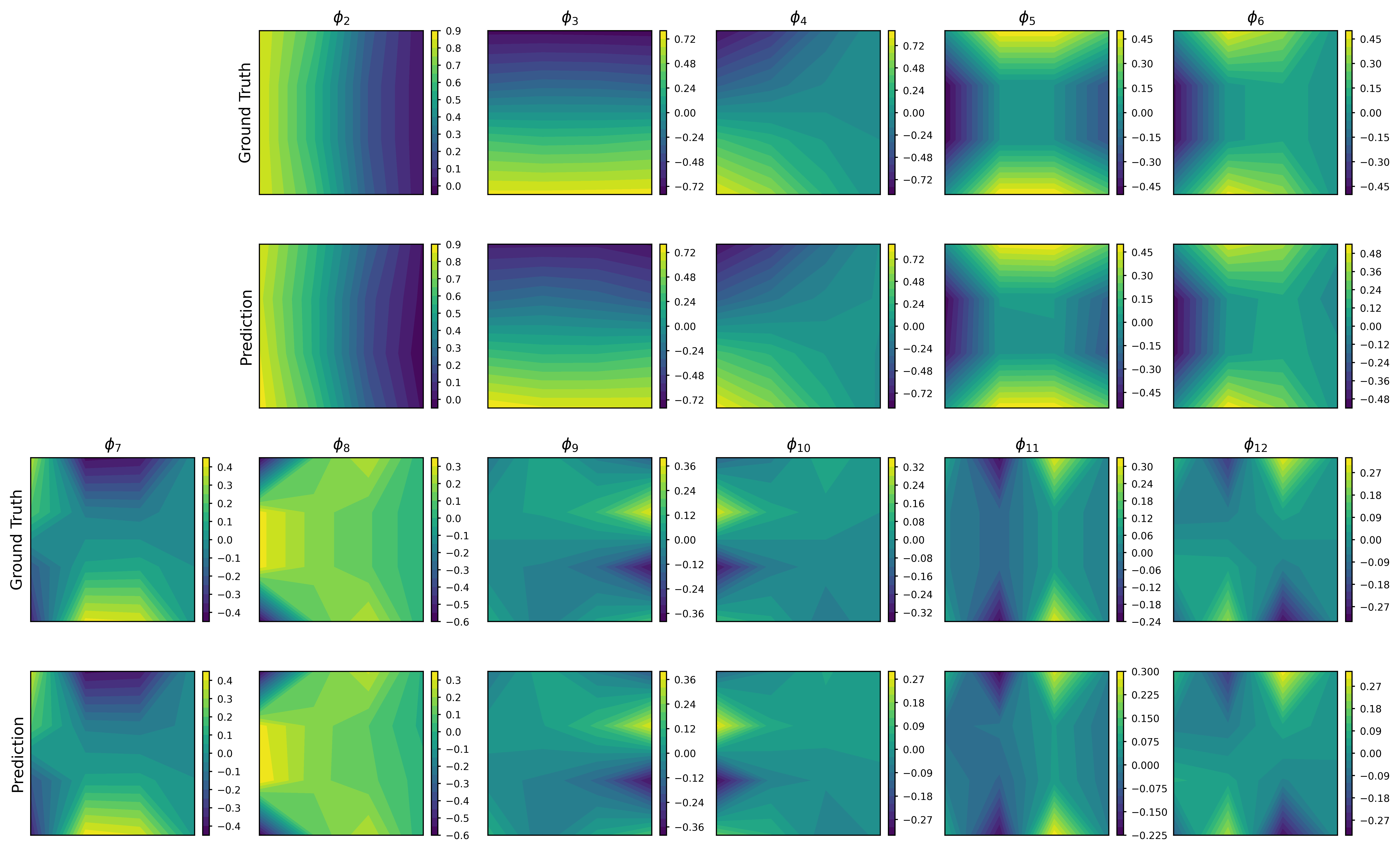}
    \caption{Schematic diagram of multiscale basis functions (2-12). For each two lines, the top denotes the ground truth, also the reference solutions, while the bottom refers to the predicted solutions.}
    \label{fig:basis_pred}
\end{figure}

\begin{figure}[ht]
    \centering
    \includegraphics[width=1.0\linewidth]{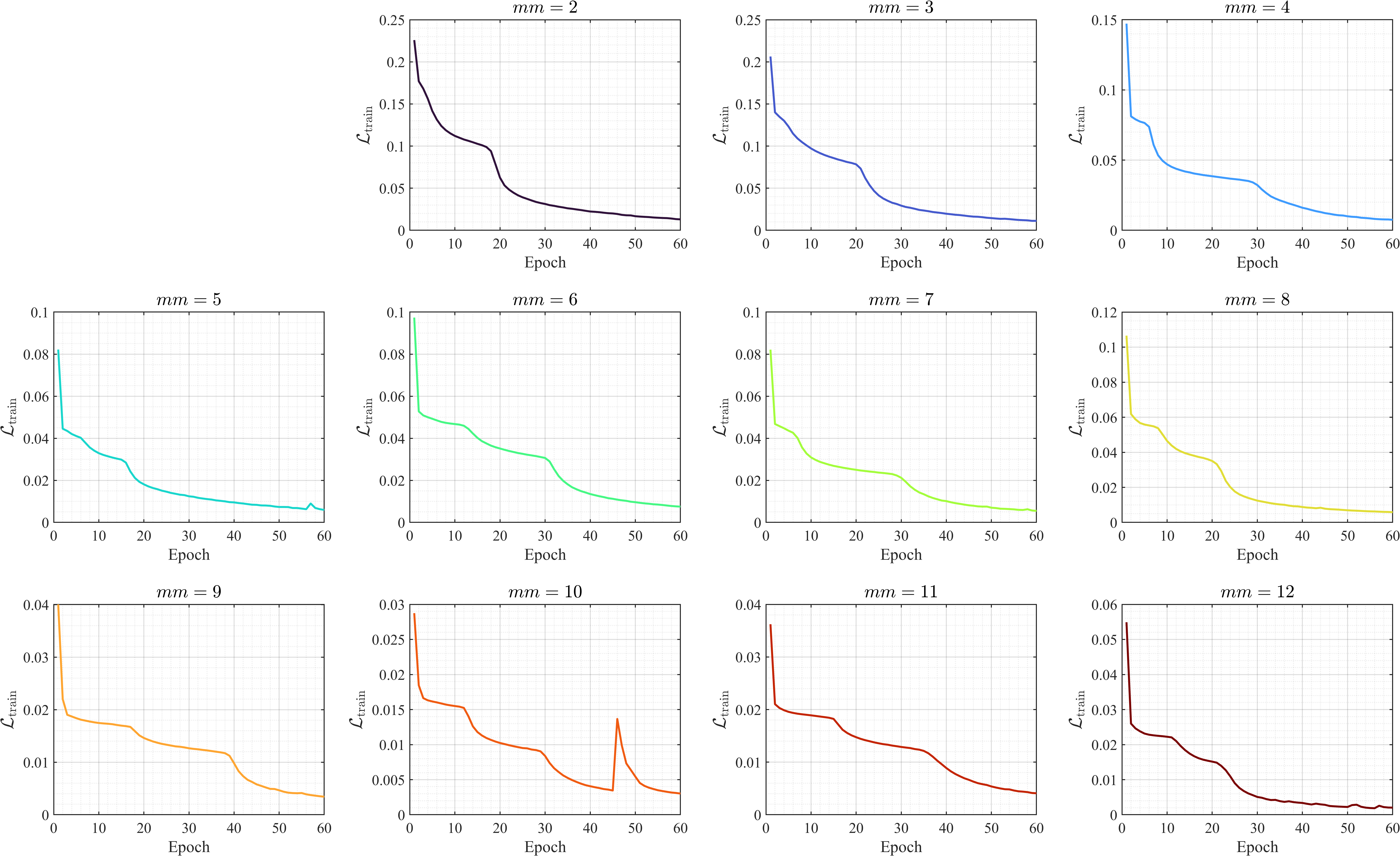}
    \caption{Learning curves of multiscale basis functions.}
    \label{fig:learning_curve}
\end{figure}

\begin{table}[ht]
  \centering
  \caption{Performance metrics for different multiscale basis functions}
  \label{tab:performance_metrics}
  \setlength{\tabcolsep}{8pt}
  \renewcommand{\arraystretch}{1.15}
  \begin{tabular}{ccc|@{\hspace{1.2cm}}ccc}
    \toprule
    \textbf{mm} & \textbf{MSE} & \textbf{R$^2$} & \textbf{mm} & \textbf{MSE} & \textbf{R$^2$} \\
    \midrule
       &        &        & 7  & 0.0030 & 0.9594 \\
    2  & 0.0056 & 0.9790 & 8  & 0.0023 & 0.9762 \\
    3  & 0.0081 & 0.9664 & 9  & 0.0013 & 0.9529 \\
    4  & 0.0079 & 0.5540 & 10 & 0.0016 & 0.9232 \\
    5  & 0.0038 & 0.9657 & 11 & 0.0014 & 0.9499 \\
    6  & 0.0028 & 0.9641 & 12 & 0.0008 & 0.9835 \\
    \bottomrule
  \end{tabular}
\end{table}

\begin{remark}
    It is worth noting that the predicted solutions in \hyperref[fig:basis_pred]{Figure~\ref{fig:basis_pred}} do not perfectly coincide with the reference solutions. This is expected, since for data-driven function approximation methods, even under large-sample conditions, the learned model generally provides only an asymptotic approximation to the true mapping rather than a perfect pointwise fit. Nevertheless, such a discrepancy is acceptable as long as it remains within a reasonable error tolerance. More importantly, the acceptability of this approximation error should not be judged solely at the level of basis-function fitting, but rather in terms of its influence on the final pressure solution. This issue, together with the corresponding tolerance criterion, will be discussed in more detail in the following sections.
\end{remark}

\begin{comment}
    % =====这里是part 2+3的手稿部分=====
    % 写mm=6的例子，不过在这之前需要写一下精度和效率的tradeoff，可以从Roff以及A的size入手，所以需要加入对应的图片（作为这部分的开头）
    回顾GMsFEM的框架，使用的多尺度函数越多，得到的逼近解越精确，但同时计算成本也越大（求解精度可见图A6）。这是因为使用的基函数数量会影响线性系统中的矩阵Roff^T BM^(-1)B^T Roff的size。以使用6个和12个基函数为例，见图5，左右两边分别展示了不同Roff和三对角矩阵的稀疏性和size，这种条件需要我们在求解精度和计算效率之间进行取舍。
    我们以mm=6为例，同时我们也提供了其他情况的算例在附录2。为说明我们的方法具有优越性，我们将结果与一些知名的模型进行对比：傅里叶神经算子（FNO）、U-Net、physics-informed neural network（PINN）和GMsFEM（mm=6），并以绝对误差与相对L2误差为评价指标进行评判。测试算例的求解见图6，同时图7中的直方图与表格展示了相对L2误差的情况。
\end{comment}

A review of the GNsFEM framework shows that using a larger number of multiscale basis functions generally leads to more accurate approximate solutions, but it also results in a higher computational cost, a illustrated bu the solution accuracy in \hyperref[fig:ms_solution_rel_L2_err]{Figure~\ref{fig:ms_solution_rel_L2_err}}. The main reason is that the number of basis functions directly affects the size of the matrix in final linear system
\begin{equation}
    R^T_{\mathrm{off}}BM^{-1}B^TR_{\mathrm{off}}p_h=R_{\mathrm{off}}^TF \notag
\end{equation}
Taking the cases with 6 and 12 basis functions as examples (\hyperref[fig:sparse_system]{Figure~\ref{fig:sparse_system}}). On the left and right, the sparsity patterns and sizes of the corresponding $R_{\mathrm{off}}$ matrices and tridiagonal matrices, respectively. This observation indicates that a trade-off must be made between solution accuracy and computational efficiency in practical computations. 

\begin{figure}[!h]
    \centering
    \includegraphics[width=0.7\linewidth]{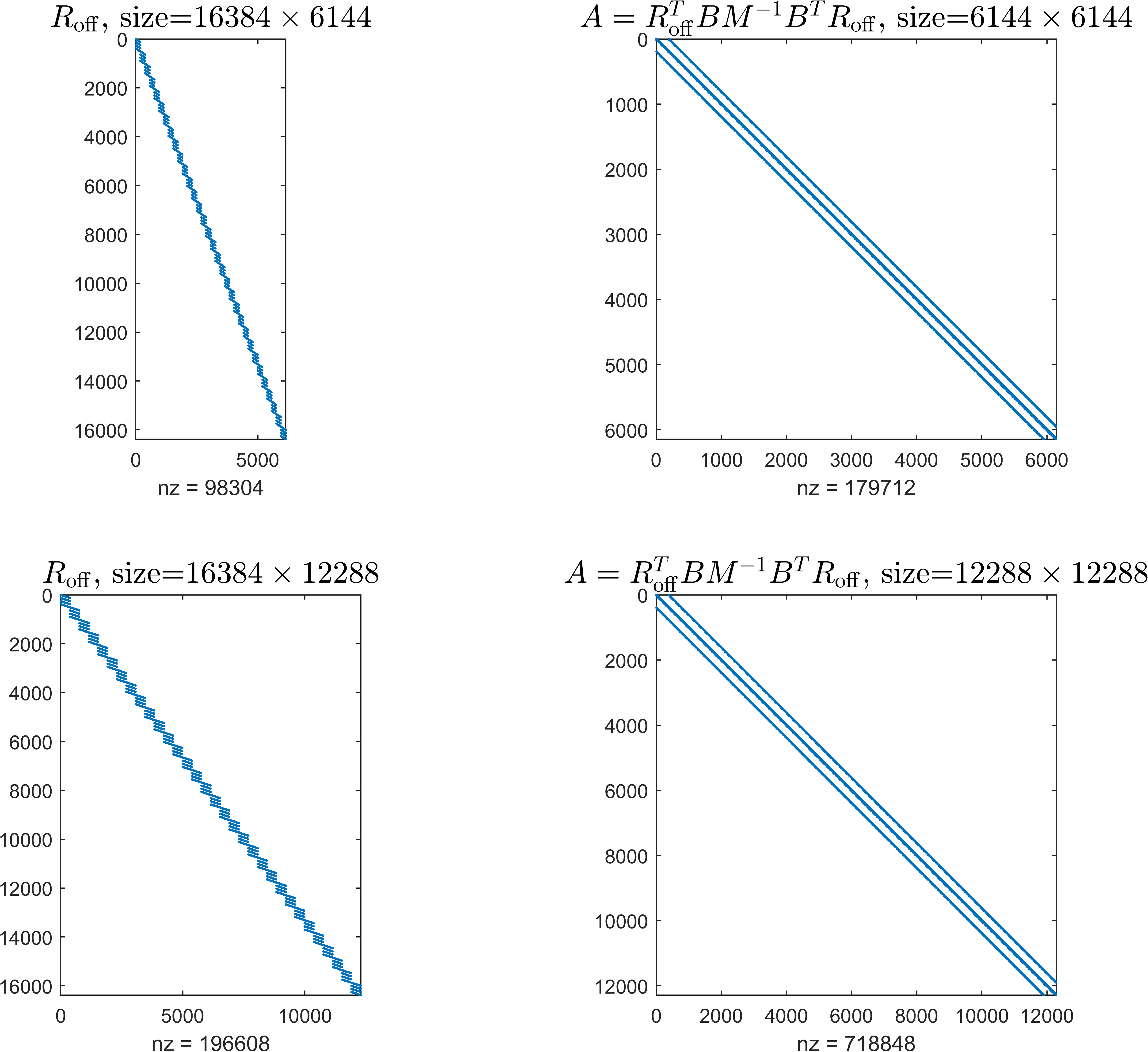}
    \caption{Schematic diagrams of sparse matrix in the linear system $Ap=F$. Top: case using 6 basis functions; bottom:case using 12 basis functions. $nz$ refers to the number of non-sparse points.}
    \label{fig:sparse_system}
\end{figure}

In this work, we mainly consider the case $mm=6$, while additional numerical results for other settings are provided in \hyperref[appendix:pre_sol]{\ref{appendix:pre_sol}}. To demonstrate the superiority of the proposed method, we compare it with several well-known models, including FNO \cite{li2020fourier}, U-Net \cite{liu2025learning}, physics-informed neural network (PINN) \cite{liu2025adaptive,rassi2019physics}, and GMsFEM ($mm=6$) \cite{chen2020generalized}. Absolute error and relative $L2$ error are adopted as the evaluation metrics to assess the solution performance of different methods. The solution results for the test examples are shown in \hyperref[fig:ref_baseline]{Figure~\ref{fig:ref_baseline}}, while the histogram and tables in \hyperref[fig:rel_L2_baseline]{Figure~\ref{fig:rel_L2_baseline}} further present the relative $L2$ error statistics.

\begin{figure}[h]
    \centering
    \includegraphics[width=1.0\linewidth]{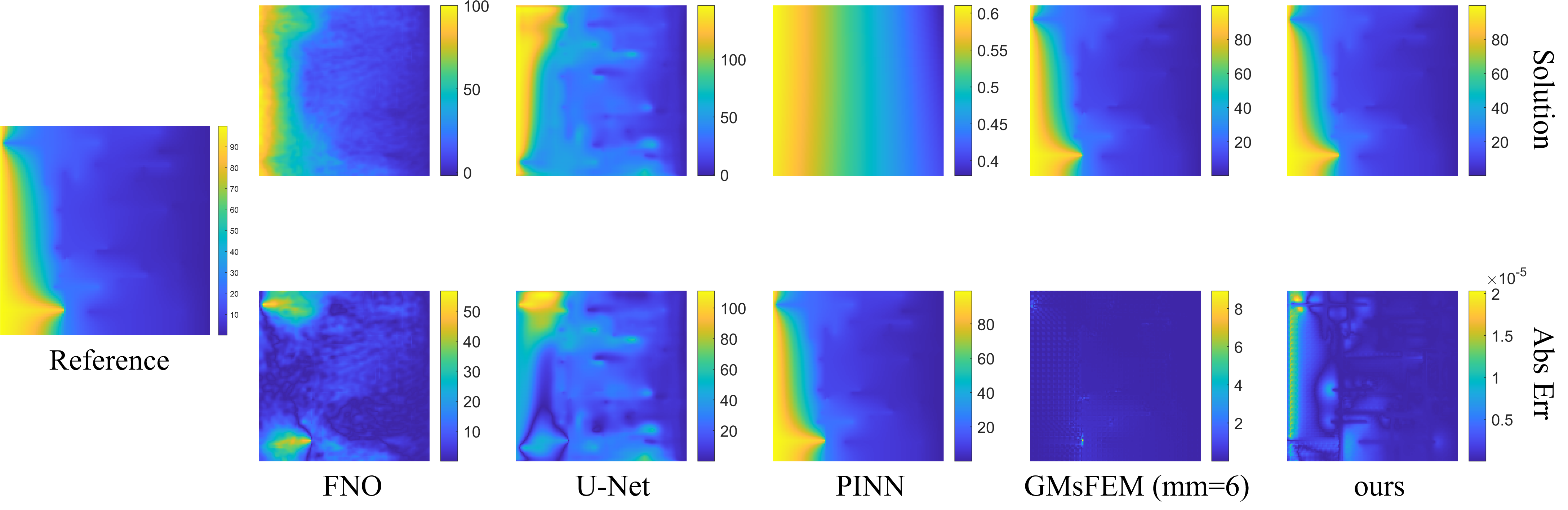}
    \caption{Comparison of our model and baseline models. Left: reference solution; Right Top: pressure solutions; bottom: abstract error. In this case, 'ours' also refers to that we use $mm=6$ to solve the linear system.}
    \label{fig:ref_baseline}
\end{figure}

\begin{figure}[h]
    \centering
    \includegraphics[width=1.0\linewidth]{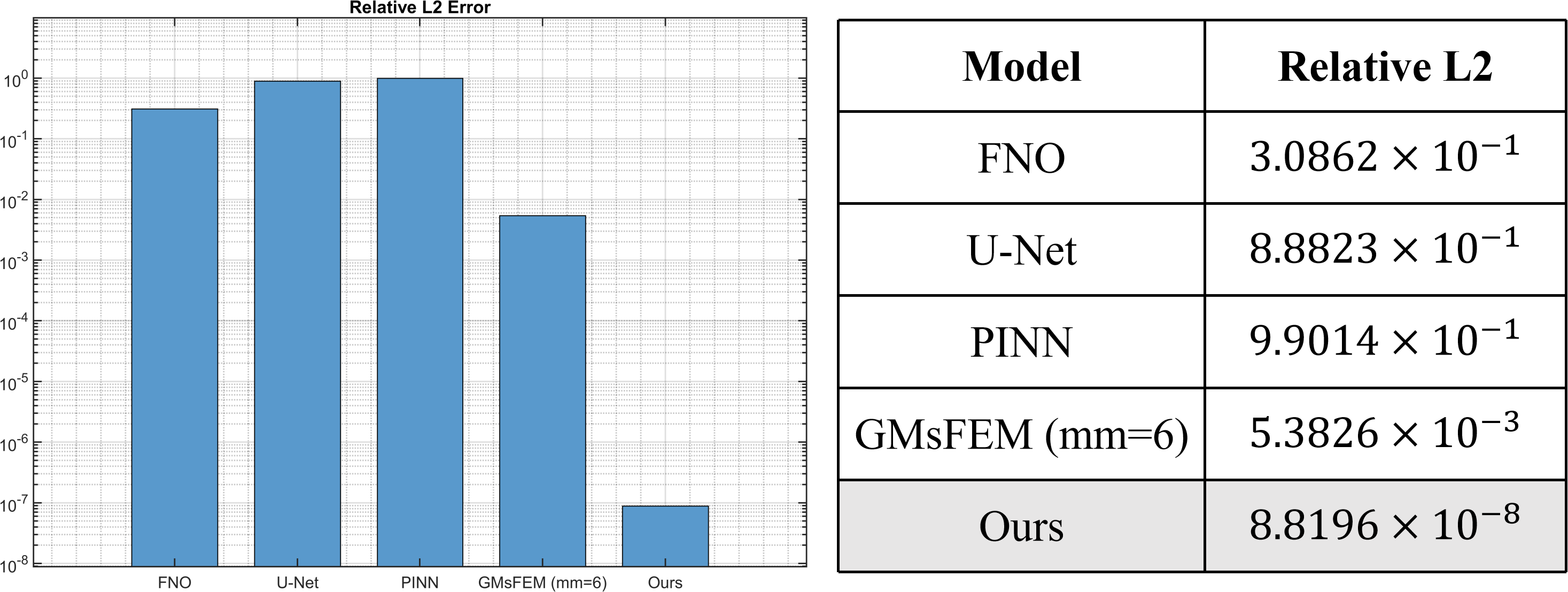}
    \caption{Relative L2 error of ours and baseline models}
    \label{fig:rel_L2_baseline}
\end{figure}

\begin{comment}
    综合以上结果，我们认定我们的方法在Darcy方程的求解上具备最优的准确性。FNO和U-Net能够学习到一定的空间特征，但是对于细节部分并没有能力进行学习。PINN在强异质性、高对比度系数的条件下对物理场的学习出现了明显的困难。而GMsFEM虽然也能实现比较理想的精度，但由于其重复的局部计算带来了明显的计算成本的增加，虽然数值方法在可解释性上具备优越性，与我们的方法相比，在高分辨率情景下并无明显优势。
\end{comment}
Based on the above results, we conclude that the proposed method achieves the best overall accuracy for solving the Darcy equation. The comparison results show that, although FNO and U-Net are capable of learning certain spatial features, their ability to capture local details and complex structures remains limited. PINN exhibits clear difficulty in learning the physical field under strongly heterogeneous and high-contrast coefficient settings, which leads to relatively large discrepancies between its predictions and the reference solutions. Although GMsFEM can also achieve reasonably good accuracy, it inherently relies on repeated local offline computations, resulting in a noticeably higher computational cost. Therefore, while traditional numerical methods possess a natural advantage in interpretability, they do not show a clear overall advantage over the proposed method in high-resolution scenarios.

In addition, the complexity analysis shows that, for an input size of $1\times128\times128$, the proposed model requires 7.6631 GFLOPs for a single forward pass. The main computational cost is concentrated in the high-resolution convolutional layers of the U-Net decoder, while the additional overhead introduced by the Fourier branch and attention modules remains relatively limited. Overall, the proposed method maintains high predictive accuracy while keeping the computational complexity within an acceptable range, thereby achieving a favorable balance between accuracy and efficiency.

%% file: section7_conclusion.tex
\section{Conclusion and Future Research}
This work presents a hybrid framework for multiscale Darcy solving in which learning is used only for the rapid prediction of multiscale basis functions, while the subsequent global assembly and pressure solve are still carried out by a two-grid preconditioner. The numerical results show that the framework improves final pressure reconstruction relative to several representative learning-based methods and remains stable under strongly heterogeneous and high-contrast coefficient settings. In comparison with traditional GMsFEM, its main advantage lies in reducing the computational cost of the basis-construction stage rather than replacing the subsequent global numerical solution procedure. By retaining the two-grid solver, the proposed approach improves the efficiency of basis generation while preserving the stability and interpretability of the underlying numerical framework. Overall, these results suggest that accelerating multiscale basis construction through learning, while combining it with a mature numerical solver, provides a viable strategy for high-resolution Darcy-type problems. Future work will consider extensions to three-dimensional settings and more complex models, together with a more rigorous theoretical analysis linking basis-prediction error to global solution error.

%% file: sec_supp_info.tex
\section*{Acknowledgment}
The research of Peiqi Li is supported by the Postgraduate Research Scholarship of Xi'an Jiaotong-Liverpool University (FOSA2412003).

\section*{Credit Author Contribution Statement}
\textbf{Peiqi Li}: Data Curation, Methodology, Project Administration; Validation, Visualization, Writing-original draft; \textbf{Jie Chen}: Supervision, Writing-review; \textbf{Shubin Fu}: Methodology, Proof reading.

\section*{Declaration of conflicting interests}
The authors declared no potential conflicts of interest with respect to the research, authorship, and publication of this article.

%% file: appendix_multiscale_solution.tex
\section{GMsFEM-based Solutions}
\label{appendix:ms_sol}
\begin{figure}[!htbp]
    \centering
    \includegraphics[width=1.0\linewidth]{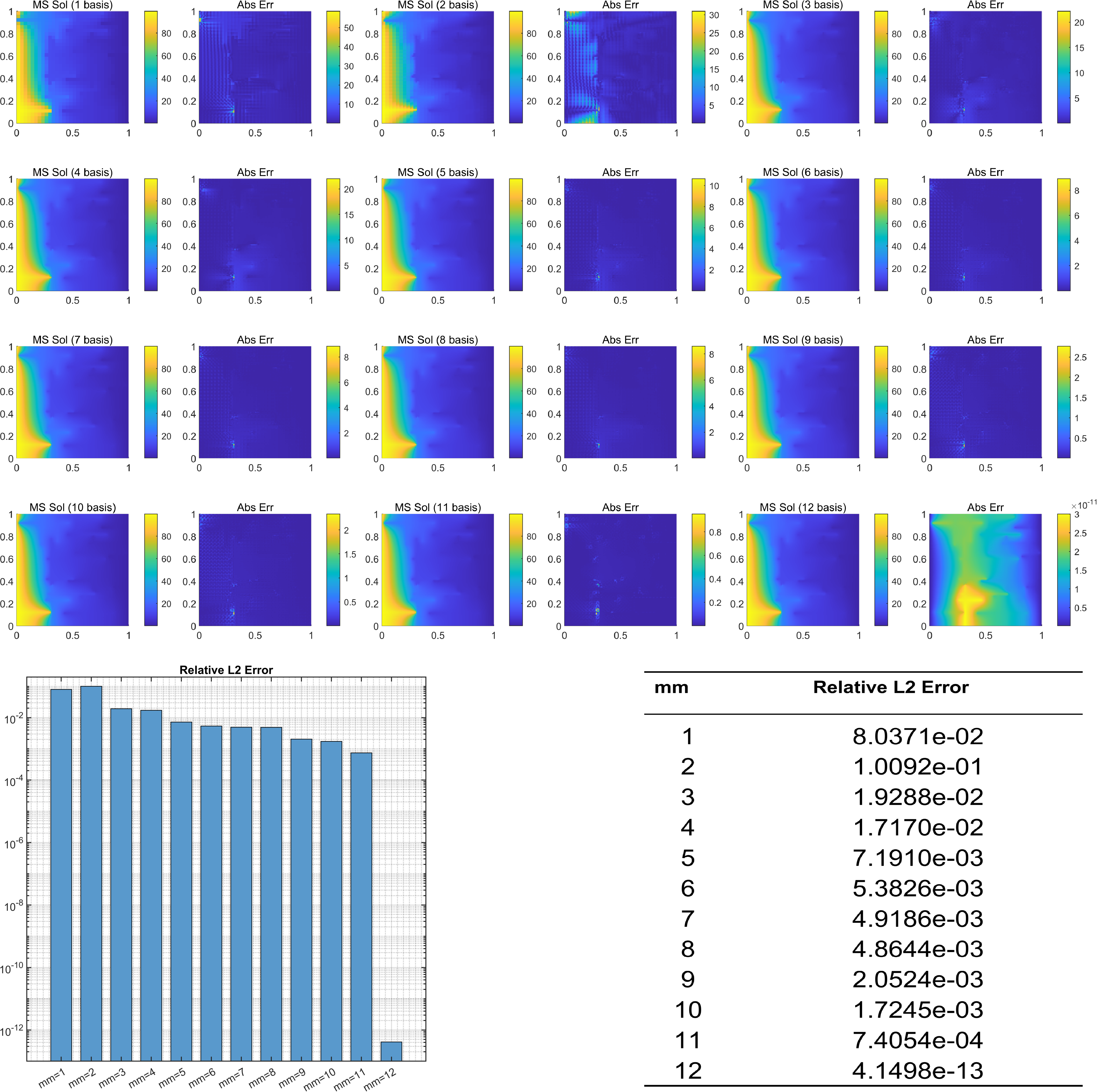}
    \caption{Schematic diagram of the solution based on GMsFEM. Upper part: Pressure solutions using different numbers of multi-scale basis functions (which we refer to as multiscale solutions) and the corresponding absolute errors. Lower part, left side: Histogram of relative L2 error; Bottom right part: Specific values of relative L2 error. MS Sol: Multiscale Solution; Abs Err: Absolute Error; $mm$: Number of used multiscale basis functions.}
    \label{fig:ms_solution_rel_L2_err}
\end{figure}

%% file: appendix_preconditioned_solutions.tex
\newpage
\section{Predicted Solutions Using Different Numbers of Multiscale Basis Solutions}
\label{appendix:pre_sol}

\begin{figure}[!h]
    \centering
    \includegraphics[width=1.0\linewidth]{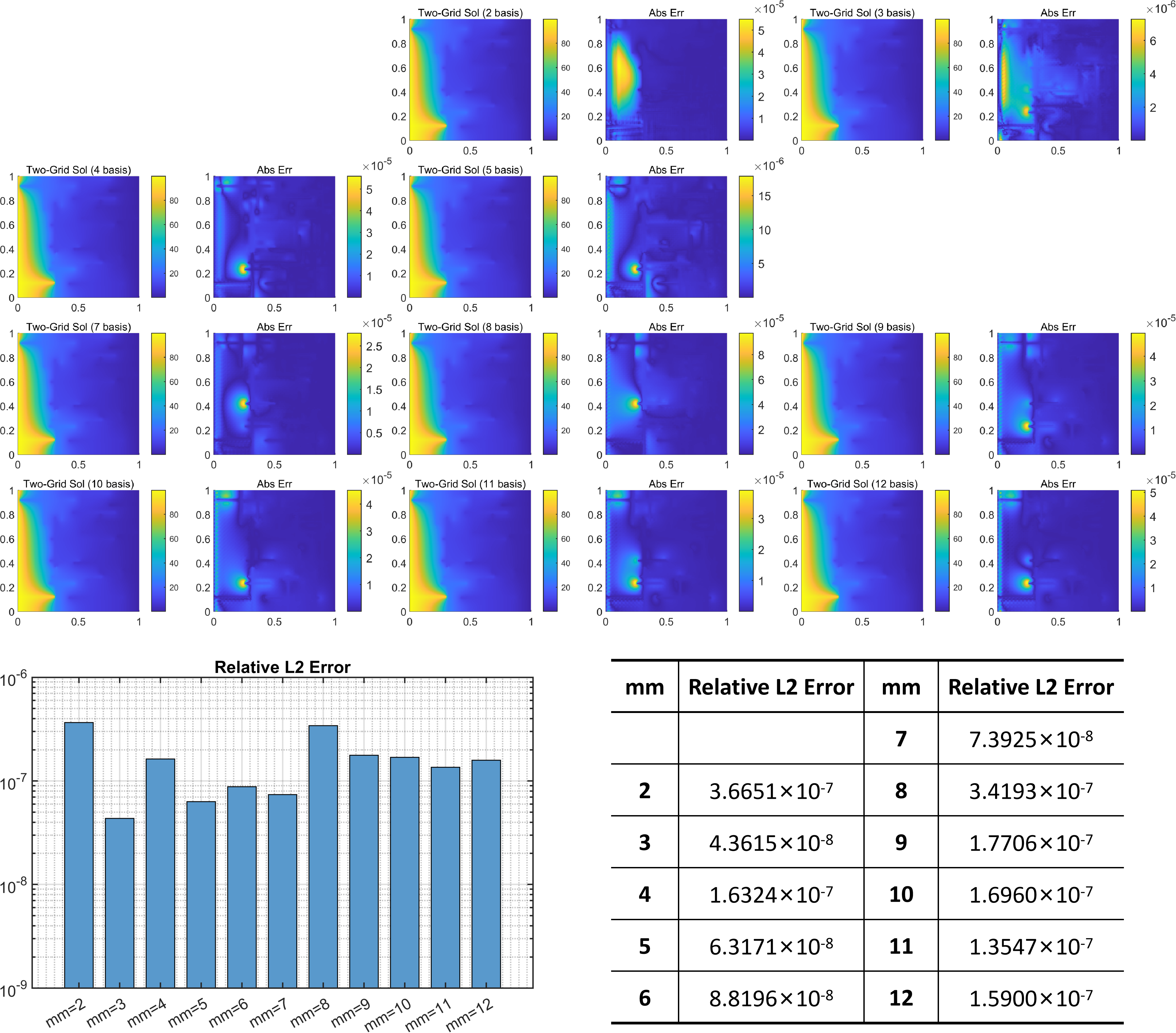}
    \caption{Schematic diagram of the solutions of iterative two-grid algorithm using $\hat{R}_{\mathrm{off}}$. Upper part: Pressure solutions using different numbers of predicted multi-scale basis functions in two-grid preconditioner and the corresponding absolute errors. Lower part, left side: Histogram of relative L2 error; Bottom right part: Specific values of relative L2 error. Abs Err: Absolute Error; $mm$: Number of used multiscale basis functions. In this figure, the part of '$mm=1$' means using only one basis function that is not learning using networks, so it keep the same as the one in \hyperref[appendix:ms_sol]{\ref{appendix:ms_sol}}. Meanwhile, the part '$mm=6$' is shown in \hyperref[sec:experiments]{Section~\ref{sec:experiments}}.}
    \label{fig:2grid_solution_rel_L2_err}
\end{figure}